\documentclass[twoside]{article}
\usepackage[a4paper]{geometry}
\usepackage{amsmath,graphicx}
\usepackage{mathtools}
\usepackage{caption}
\usepackage{subcaption}
\usepackage{tikz}
\usepackage[colorinlistoftodos,prependcaption,textsize=tiny]{todonotes}
\usepackage{wrapfig}
\usepackage{hyperref}
\usepackage{cleveref}
\usepackage{thm-restate}
\usepackage{comment}
\usepackage{array}
\usepackage{amsmath}
\newcommand\numberthis{\addtocounter{equation}{1}\tag{\theequation}}
\usepackage[charter]{mathdesign}

\tikzset{
    cross/.pic = {
    \draw[rotate = 45] (-#1,0) -- (#1,0);
    \draw[rotate = 45] (0,-#1) -- (0, #1);
    }
}
\usepackage[colorinlistoftodos,prependcaption,textsize=tiny]{todonotes}

\usepackage[boxruled,vlined]{algorithm2e}
\usetikzlibrary{positioning}
\usepackage{comment}

\usetikzlibrary{positioning}
\usepackage[charter]{mathdesign}
\usetikzlibrary{shapes.symbols}
\usetikzlibrary{calc}

\newtheorem{theorem}{Theorem}[section]
\newtheorem{lemma}[theorem]{Lemma}

\newtheorem{definition}[theorem]{Definition}
\newtheorem{remark}[theorem]{Remark}

\newenvironment{proof}{{\bf Proof:}}{\hfill\rule{2mm}{2mm}}

\newcommand{\TR}{\textsc{Trapezoid}}
\newcommand{\BA}{\textsc{Ball}}
\newcommand{\LL}{\mathcal{L}}

\newcommand{\EL}{\ell}
\newcommand{\HIT}{\textsc{HitSet}}
\newcommand{\EP}{\epsilon}

\newcommand{\EO}{e_1^*}
\newcommand{\ET}{e_2^*}
\newcommand{\FT}{\textsc{Ft}}
\newcommand{\QFT}{\textsc{Ft-Query}}
\newcommand{\TL}{\tilde{O}}
\newcommand{\DIA}{\diamond}

\newcommand{\DD}{\mathbb{D}}
\newcommand{\PP}{\mathbb{P}}
\newcommand{\QU}{\textsc{Query}}
\newcommand{\AN}{\textsc{Ans}}

\newcommand{\DO}{\textsc{Do}}
\newcommand{\SDO}{\textsc{Sdo}}

\newcommand{\LCA}{\textsc{Lca}}
\newcommand{\NW}{(nW)}
\newtheorem{assumption}[theorem]{Assumption}

\usetikzlibrary{decorations.pathreplacing,calligraphy}
\usetikzlibrary{
    decorations.markings,
    decorations.pathmorphing,
    decorations.text,
}
\title{Near Optimal Dual Fault Tolerant Distance Oracle}

\author{
Dipan Dey \\
  IIT Gandhinagar \\
  Gandhinagar \\
  India\\
  \texttt{dey\_dipan@iitgn.ac.in} \\
  \and
  Manoj Gupta \\
  IIT Gandhinagar \\
  Gandhinagar \\
  India\\
  \texttt{gmanoj@iitgn.ac.in} \\
}

\date{}
\begin{document}

\maketitle

%-Title

%\received{}
%\received[revised]{}
%\received[accepted]{}

\maketitle
%+Abstract
\begin{abstract}
	We present a dual fault-tolerant distance oracle for undirected and unweighted graphs. Given a set $F$ of two edges, as well as a source node $s$ and a destination node $t$, our oracle returns the length of the shortest path from $s$ to $t$ that avoids $F$ in $O(1)$ time with a high probability. The space complexity of our oracle is $\Tilde{O}(n^2)$ \footnote{$\Tilde{O}$ hides poly$\log n$ factor }, making it nearly optimal in terms of both space and query time.

Prior to our work, Pettie and Duan [SODA 2009] designed a dual fault-tolerant distance oracle that required $\Tilde{O}(n^2)$ space and $O(\log n)$ query time. In addition to improving the query time, our oracle is much simpler than the previous approach.

\end{abstract}

\newpage

\section{Introduction}
Graph theory is a valuable tool for modelling and solving many real-world problems. As real-world networks are susceptible to failures, such as faulty edges or vertices within the network, we must design a system that continues to operate in the presence of failures. In this paper, we focus on developing algorithms that effectively operate in the presence of faults. Consider the scenario where we aim to determine the shortest path from a source vertex to a destination vertex. However, we become aware that one of the edges in the graph is unavailable or faulty due to some reasons. Under these circumstances, our objective becomes finding an alternative shortest path that avoids the faulty edge.

Let us describe an abstract model in which we want to solve the above problem. We consider an undirected and unweighted graph, denoted as $G$. To facilitate efficient query processing, we preprocess this graph to construct a suitable data structure. This data structure is used to answer queries about shortest paths in $G$ in the presence of faults.
We assume that we will receive queries of the following form:
\vspace{-0.15 cm}
\begin{center}
$\QU (s,t,F)$: Find length of the shortest path between vertices $s$ and $t$ in graph $G$, avoiding set of edges $F$.
\end{center}
\vspace{-0.2 cm}
The objective of our algorithm is to efficiently answer the above query by utilising the prepared data structures. After creating a data structure, while queried about $\QU(s,t,F)$, we run an algorithm (named query algorithm) to answer the query. This combination of data structure and query algorithm is commonly referred to as an {\em oracle}. Our paper focuses on calculating distances in the presence of faults. Consequently, we refer to our oracle as a {\em fault-tolerant distance oracle}.

We evaluate the fault-tolerant distance oracle based on two key parameters: (1) the space required by the data structure and (2) the query time, which refers to the time taken to respond to a query.
To facilitate the description of our results in this context, we introduce a simple notation. We define a distance oracle that avoids a fault set of size at most $f$ as $\DO(f)$. When the source vertex $s$ is fixed in the  $\QU$, the oracle is referred to as a Single Source Distance Oracle or $\SDO(f)$ for brevity.

In this paper, our focus is designing  $\DO(2)$ -- i.e., designing fault-tolerant distance oracle that can handle fault set of size at most 2.
We discuss the previous results obtained in this area in \Cref{table}.

\begin{table}[hpt!]
\begin{tabular}{|c|c|c|p{4 cm}|c|}
\hline
Oracle & Space &Query time & Remarks & Ref\\
\hline
$\DO(1)$ & $\TL(n^2)$ & $O(1)$ &-&\cite{Demetrescu2008}\\
\hline
$\DO(2)$ & $\TL(n^2)$ & $O(\log n)$ &-&\cite{Duan2009}\\
\hline
$\DO(f)$ & $\TL(n^{3- \alpha})$ &  $\TL(n^{2-(1- \alpha)/f })$ & $\alpha \in [0,1]$ when the preprocessing time  is $O(Mn^{3.376 - \alpha})$ and  edge weights are integral in the range $[-M, \dots , M].$ & \cite{10.1145/2438645.2438646}    \\
\hline

$\DO(f)$ & $\TL(n^{2+\alpha})$ & $O(n^{2-\alpha}f^2 + n f^{\omega})$ & $\alpha \in [0,1]$ and $\omega $ is the matrix multiplication exponent \cite{coppersmith1987matrix,stothers2010complexity,Williams12,Gall14a, alman2021refined}& \cite{BrandS19}\\

\hline
$\DO(f)$ & $O(fn^4)$ & $O(f^{O(f)})$ &-&\cite{DuanR22}\\
\hline
$f$ & $O(n^2)$ & $\TL(nf^{\omega})$ & Edge weights are in the range $[1\dots  W]$ & \cite{KarczmarzS23} \\
\hline
$f$ & $O(f^4 n^2 \log \NW)$ & $O(c^{(f+1)^2}f^{8(f+1)^2}\log^{2(f+1)^2} \NW)$ & Edge weights are in the range $[1\dots  W]$ & \cite{DeyGupta24} \\
\hline
\end{tabular}
\caption{Relevant results for fault-tolerant distance oracles}
\label{table}
\end{table}

For $\DO(2)$, the oracle of Pettie and Duan \cite{Duan2009} is nearly optimal. It takes $\TL(n^2)$ space and $O(\log n)$ query time. However, as the authors of  \cite{Duan2009} also mentioned, their approach is overly complicated and requires extensive case analysis. Duan and Ren \cite{DuanR22} proposed a fundamentally different approach. They designed an algorithm for $\DO(f)$ that takes $O(n^4)$ space and $O(f^{O(f)})$ query time. For $\DO(2)$, their query time is $O(1)$. Unlike \cite{Duan2009}, the main feature of their algorithm is that it involves a limited amount of case analysis. However, the space taken by their algorithm is $O(n^4)$. Recently, Dey and Gupta \cite{DeyGupta24} designed an $f$-fault tolerant distance oracle with $\TL(n^2)$ space and $\TL(1)$ query time for fixed $f$. In this paper, using an approach similar to \cite{DeyGupta24} and \cite{DuanR22}, we design a dual-fault tolerant oracle with $\TL(n^2)$ size and $O(1)$ query time. We use some randomisation techniques which helps us to reduce the query time. The main result of this paper is: 

\begin{theorem}
\label{thm:main}
        For an undirected and unweighted graph, there is a dual fault-tolerant oracle that takes $\TL(n^2)$ space and answers each query in $O(1)$ time with high probability.
\end{theorem}

Unfortunately, some tools we use in this paper do not works for weighted graphs.
\vspace{-0.4 cm}
\section{Overview}
\label{sec:overview}

In this section, we will at first give an overview of the method used by Duan and Ren in \cite{DuanR22} to create an $f$-fault tolerant distance oracle with $O(n^4)$ size. After that, we will give an overview of our modifications to that approach to have an oracle for dual faults with $\TL(n^2)$ size and constant query time. For that, let us define some notations at first. The notation $st$ represents the shortest path from vertex $s$ to vertex $t$. The notation $|st|$ denotes the length of this path (same as the number of edges in the $st$ path for unweighted graphs). The notation $st \DIA F$ denotes the shortest path from $s$ to $t$ while avoiding a set of edges $F=\{e_1,e_2\}$, and $|st \DIA F|$ corresponds to the length of that path.

\subsection{Overview of the approach in \cite{DuanR22}}
Suppose we want to find the path $P= st \DIA F$. In \cite{DuanR22}, the authors introduced the concept of `maximiser', which will be discussed further in later sections. For now, we want to emphasize on the fact that the `maximisers' output either the length of the path $P$ or a set of vertices serving as a `hitting set' of $P$. We will formally prove this later. Together with these  `maximisers', they also used a structural theorem of shortest paths under edge failures to use the hitting set effectively. A path is said to be $f$-decomposable if it is the concatenation of at most $(f + 1)$ shortest paths in $G$.
In \cite{AfekBKCM02}, the authors show the following:

\begin{theorem}(Theorem 1 in \cite{AfekBKCM02}) For  any set of failures $F$ of size $f$, every shortest path from $s$ to $t$ avoiding $F$ is $f$-decomposable.
\end{theorem}

Let us now discuss an important aspect of this theorem. Let $P = st \DIA \{e_1,e_2\}$ be 2-decomposable and $x$ be a vertex of the hitting set which lies on the path $P$. Since $P$ is 2-decomposable, it contains at most three shortest paths, say $A$, $B$, and $C$, where $P$ starts with the path $A$ and ends with the path $C$. If $x$ lies on $B$, then $P[s,x]$ and $P[x,t]$ are 1-decomposable. Thus, we have hit path $P$ and reduced the problem into two smaller subproblems (in terms of decomposability). However, if $x$ lies on $A$ or $C$, then we have not reduced the problem. The whole game now is to design an algorithm that can find an $x \in B$. In \cite{DuanR22}, Ren and Duan design such an algorithm. Since our strategy is similar, let us define some notations that formalise our discussion.

\begin{definition} (First and the last segment in a path and an intermediate vertex)

Let us assume that $F$ is a set of size $f$ and $P=st \DIA F$ is $f$-decomposable. Thus, $P$ can be decomposed into at most $f+1$ shortest paths. Let $A$ and $B$ be the first and the last shortest path from these $ f+1$ shortest paths. Thus, $A$ is the shortest path that starts from $s$, and $B$ is the shortest path that ends at $t$. $A$ is called  the first segment of the path $P$ and $B$ is called the last segment of the path $P$.  If a vertex $x \in P$ does not lie in the first and last segments, then it is called an intermediate vertex.
\end{definition}

Before we describe further, let us at first define some notations. 

\begin{definition}
A path $sx$ is said to be {\bf intact from failures} $\{e_1,e_2\}$ if the endpoints of $e_1$ and $e_2$ do not lie strictly between\footnote{\cite{DuanR22} did not require the endpoints to be strictly inside. However, we require it to avoid some corner cases.} $s$ and $x$ on the path $sx$. Similarly, if the subtree under a node $x$ in $T_s$, that is $T_s(x)$, do not contain any endpoints of $e_1$ and $e_2$, it is intact from failures $\{e_1,e_2\}$.
\end{definition}

The above definition may appear natural. However, it is not. Usually, when we say that a path or tree is intact from failures $\{e_1,e_2\}$, we mean that it does not contain $e_1$ and $e_2$. In contrast, the above definition does not allow the endpoints of $e_1$ and $e_2$. Though never emphasised, this fact is used crucially in the algorithm of Ren and Duan \cite{DuanR22}.
We now define the notion of a {\em clean} vertex introduced in \cite{DuanR22}.
\begin{definition}
A vertex $x$ is called $s$-clean from $\{e_1,e_2\}$ if $sx$ and $T_s(x)$ are intact from failures $\{e_1,e_2\}$.
\end{definition}

\begin{figure}[hpt!]
\centering
\begin{tikzpicture}[node distance=3.5cm, box/.style={draw, rectangle, minimum width=2cm, minimum height=1cm}]

% First Box
\draw (0,0)% circle (1.5cm)
    node[align=center,minimum size=3cm,draw,box] (box1)
    {Found a \\$s$-clean vertex};
% Second Box
\draw (0,0)% circle (1.5cm)
    node[right of=box1,align=center,minimum size=3cm,draw,box] (box2)
    {Found a \\$t$-clean vertex};

% Third Box
\draw (0,0)% circle (1.5cm)
    node[right of=box2,align=center,minimum size=3cm,draw,box] (box3)
    {Found an \\intermediate vertex};
% Arrow
\draw[->] (box1) -- (box2);
\draw[->] (box2) -- (box3);

\node[align=center,minimum size=0.5cm] at (0,-1)
    {After the first call\\ to the maximizer};
\node[align=center,minimum size=0.5cm] at (3.6,-1)
    {After the second call\\ to the maximizer};
\node[align=center,minimum size=0.5cm] at (7,-1)
    {After the third call\\ to the maximizer};
\end{tikzpicture}
\caption{After at most three calls to the maximizer, we are sure to find an intermediate vertex.}
\label{fig:approachDuanR}
\end{figure}
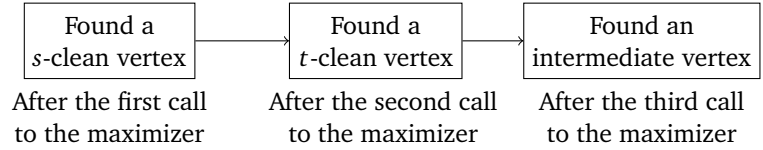
Duan and Ren used the above `clean' vertices crucially in \cite{DuanR22}. Let \(P = st\DIA F\) be \(f\)-decomposable. As mentioned above, we want to find an intermediate vertex \(x\) on \(P\), after which we can recurse on \(P[s,x]\) and \(P[x,t]\). To this end, \cite{DuanR22} designed an algorithm $\QU$, which is initially called using the parameters \((s,t, F,f)\). The main component of this algorithm is the function \(\textsc{HitSet}\), a function that uses maximisers. \textsc{HitSet} outputs \(O(\text{poly}(f))\) edges \(H\) and an upper bound on $|P|$, which is stored in \(L\). The function \(\textsc{HitSet}\) has the following properties: (1) Either \(L = |P|\) or (2) One of the endpoints of the edges in \(H\) hits \(P\) and is an intermediate vertex.
 Thus, we either find our answer or recurse on two parts in Line 6 of Algorithm 1. One can check that if the running time of \textsc{HitSet} is polynomial in \(f\), then the running time of the query algorithm is \(O(f^{O(f)})\). In our case, this is \(O(1)\) since \(f = 2\).

\vspace{-0.2 cm}
\begin{algorithm}
\caption{Query Algorithm for the Exact Distance Oracle: \textsc{Query}$(s,t,F, f)$}
\label{alg:query}
\textbf{if} $st \cap \{e_1,e_2\} = \emptyset$ \textbf{return $|st|$}

\textbf{if} $f=0$ \textbf{return $\infty$}

$(L, H) \leftarrow \textsc{HitSet}(s, t, F)$\;

$\AN \leftarrow L$\;
\For{each $x$ in $H$}{
$\AN \leftarrow \min\{\text{ans}, \textsc{Query}(s, x, F, f-1) + \textsc{Query}(x, t, F, f-1)\}$;
}
\Return $\AN$;
\end{algorithm}
\vspace{-0.25 cm}

We now give an overview of how to find an intermediate vertex. Let us assume the worst case when the output of our maximisers always hit the path \(P\). The key to finding an intermediate vertex is first finding a clean vertex from the source and the destination, i.e., an \(s\)-clean and \(t\)-clean vertex. The authors first show how to design an appropriate maximiser that ensures we obtain a clean vertex. Once a \(s\)-clean and \(t\)-clean vertex are found, authors design an appropriate maximiser that gives us an intermediate vertex. This approach is succinctly described in the \Cref{fig:approachDuanR}. This figure illustrates behaviour of the \textsc{HitSet} algorithm in the worst case. In the best case, the first call to the maximiser may give us an intermediate vertex.
%However, in the worst case, the first call may give us an \(s\)-clean vertex, %and the second give us a \(t\)-clean vertex. After that, we are sure %to find an intermediate vertex.

Though we have hidden many technical details, we want the reader to note the steps required to find an intermediate vertex in the worst case. Also, to use their maximisers, they needed $O(n^4)$ space which we reduced in this paper for $2$ faults. 

\subsection{Overview of our approach}

We define some new notations which will help us in describing our approach. Suppose we want to find $st \diamond F$ where $F=\{e_1,e_2\}$.

\begin{definition}(Primary and Secondary path)
\label{def:primsecond}

The shortest $st$ path is referred to as the primary path, which includes the edge $e_1$. The secondary path is $st \DIA e_1$.
\end{definition}

\begin{figure}[hpt!]
\centering
\begin{tikzpicture}[scale=0.5]
\coordinate (s) at (-6,8);
\coordinate (t) at (2,8);
\coordinate (u) at (0,8);
\coordinate (v) at (3,8);
\coordinate (e) at (0,6);
\filldraw [black] (s) circle (3pt);
\filldraw [black] (t) circle (3pt);

\draw[ line width=0.7mm,black,opacity=0.5] (s) node[left][black] () {$s$} -- (t) node[right][black] () {$t$};
%\filldraw [red] (s) circle (2pt);
%\filldraw [red] (t) circle (2pt);
%\filldraw [black] (-2.5,8) circle (3pt);
%\filldraw [black] (-0.5,9.6) circle (3pt);
\draw (-2.5,8) node [below][scale=0.8] {$e_1$};
\draw (-0.65,9.6) node [left][scale=0.8] {$e_2$};

\draw  (u) arc(0:180:2.5) [ line width=0.7mm,blue,opacity=0.5];
%\draw  (v) arc(0:126:2.5) [ line width=0.7mm,teal,opacity=1];

%\draw (-2.5,8)[red,line width=0.5 mm] pic[rotate = 0] {cross=4pt};
%\draw (-0.55,9.6)[red,line width=0.5 mm] pic[rotate = 30] {cross=4pt};
%\draw (u) node [left] {$u$};

\filldraw [red] (-3,8) circle (2pt); % Circle to the left of e1
%\draw (-3,8) node [below] {$a$};
\filldraw [red] (-2,8) circle (2pt); % Circle to the right of e1
\draw[ line width=0.7mm,red,opacity=1] (-3,8) node[above][black][scale=0.8] () {$a$} -- (-2,8) node[above][black][scale=0.8] () {$b$};

%\draw (-2,8) node [below] {$b$};

\filldraw[red] (-0.7,9.75) circle (2pt); % Point to the left of e2 on the arc
%\draw(-0.5,9.6) ++(300:0.4)node [right] {$d$};
\filldraw[red] (-1.3,10.2) circle (2pt); % Point to the right of e2 on the arc

%\draw(-0.5,9.6) ++(131:0.6)node [right] {$c$};
\draw[ line width=0.7mm,red,opacity=1] (-1.3,10.2) node[above][black][scale=0.8] () {$c$} .. controls (-0.89,9.9) .. (-0.7,9.75) node[right][black][scale=0.8] () {} ;
\draw (-0.7,10) node [right][black][scale=0.8] {$d$};
\end{tikzpicture}
\caption{The primary path contains $e_1$ and the secondary path contains $e_2$. }
\label{fig:primarysecondary}
\end{figure}
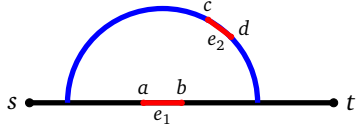

Following the above definition, we assume that the primary path contains $e_1$ (without loss of any generality). If the primary path contains none of the edges $e_1$ and $e_2$, then finding $st \DIA \{e_1,e_2\}$ is trivial. If the primary and the secondary path do not contain $e_2$, then we can use the single edge fault-tolerant data structure of Demestrescu et al. \cite{Demetrescu2008} to find $|st \DIA e_1|$. This is also an easy case. For the rest of the paper, we will assume that we are in one of the two cases:
\begin{enumerate}
\item  $e_1$ and $e_2$ both lie on the primary path 
\item  $e_1$ lies on the primary path and $e_2$ lies on the secondary path. \end{enumerate}

To apply our algorithm, we should quickly identify in which case we are operating. For this, we will introduce small weights to our originally unweighted graph similar to \cite{Bernstein2009, Hershberger2001, Parter2013, GuptaS18}. This will ensure unique shortest paths between any two vertices in the graph, even after the occurrence of two faults. Let us abuse notation and use  $|st|$ as the weight of the shortest path in this weighted graph. Also, let us define $w(e)$ as the weight of edge $e$ in this new graph. To find if $e_1 $ lies on the $st$ path, we just need to check if $|st| = |se_1| + w(e_1) + |e_1t|$. Note that we can fetch these quantities in $O(1)$ time if we have stored the shortest path tree from each vertex in the graph. So, we can check if both $e_1$ and $e_2$ lie on the $st$ path in $O(1)$ time. Similarly, we can find if $e_2$ is in secondary path  by checking is $|st \DIA e_1| = |se_2 \DIA e_1| + w(e_2) +\ |e_2t\DIA e_1|$.  Again, all the quantities can be found in $O(1)$ time using the one fault data structure of \cite{Demetrescu2008}. Henceforth, we will assume that we know the case in which we are operating.

We will show later that case (1) (both faults lie on the primary path) is similar to one of the subcases of the case (2) ($e_1$ on the primary path and $e_2$ on the secondary path). The bulk of our paper will be devoted to case (2). Henceforth, we will always assume the following setting: we want to find the shortest path from $s$ to $t$ avoiding $e_1=(a,b)$ and $e_2=(c,d)$. We will assume that $e_1 \in st$ and $e_2 \in st\DIA e_1$ where $a$ and $c$ are near to $s$ on $st$ and $st \DIA e_1$ paths respectively (See \Cref{fig:primarysecondary}).
Let $P = st\DIA \{e_1,e_2\}$ be a 2-decomposable path. We now define the notion of the detour of the path $P$. 

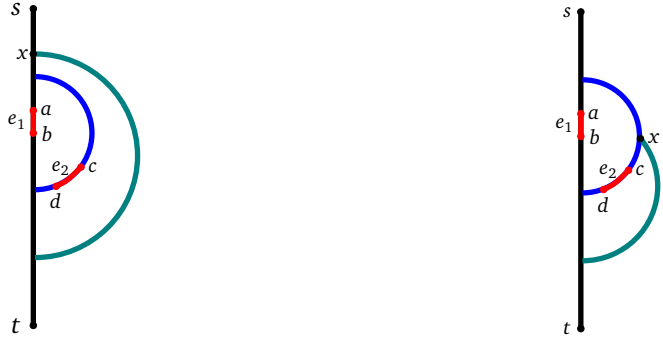
\begin{figure}[hpt!]
        \centering
        \begin{subfigure}[b]{0.4\textwidth}
        \centering
                \begin{tikzpicture}[scale=0.3]
                \coordinate (t) at (-5,-7);
\coordinate (s) at (-5,7);
\coordinate (u) at (-4.93,-1);
\coordinate (v) at (-4.93,-4);
%\coordinate (x) at (-3,0);

\draw[ line width=0.7mm,black,opacity=0.5] (t) node[left][black] () {$t$} -- (s) node[left][black] () {$s$};
%\filldraw [black][opacity=0.5] (s) circle (2pt);
%\filldraw [red] (t) circle (2pt);

%\filldraw [red] (-5,2) circle (3.5pt);

\draw (-5,2) node [left][scale=0.8] {$e_1$};
\draw (-3.8,-0.68) node [above][scale=0.8] {$e_2$};

\draw  (u) arc(-90:90:2.5) [ line width=0.7mm,blue,opacity=0.5];
\draw  (v) arc(-90:90:4.5) [ line width=0.7mm, teal ,opacity=1];
%\filldraw [red] (-4,-0.85) circle (4.5pt);
%\draw (-4,-0.85)[red,line width=0.5 mm] pic[rotate = 20] {cross=4pt};
%\draw (-5,2)[red,line width=0.5 mm] pic[rotate = 0] {cross=4pt};
\filldraw [black] (-5,5) circle(4pt);
\draw (-4.85,5) node[left][scale=0.8] {$x$};

\draw[line width=0.7mm,red,opacity=1] (-5,1.5) -- (-5,2.5);
\node[right,black][scale=0.8] at (-5,1.5) {$b$};
\node[right,black][scale=0.8] at (-5,2.5) {$a$};

\filldraw [red] (-5,2.5) circle(4pt);
\filldraw [red] (-5,1.5) circle(4pt);

\draw[ line width=0.7mm,red,opacity=1] (-4,-0.85)  .. controls (-3.5,-.6) .. (-2.91,0) ;
\node[below,black][scale=0.8] at (-4,-0.85) {$d$};
\node[right,black][scale=0.8] at (-2.91,0) {$c$};

\filldraw [red] (-4,-0.85) circle(4pt);
\filldraw [red] (-2.91,0) circle(4pt);
\filldraw [black] (s) circle (4pt);
\filldraw [black] (t) circle (4pt);

                \end{tikzpicture}
                \caption{Detour starts at  $x$  on  primary path}
        \end{subfigure}        
                \hspace{1 cm}
        \begin{subfigure}[b]{0.4\textwidth}
                \centering
                \begin{tikzpicture}[scale=0.3]

                \coordinate (t) at (-5,-7);
\coordinate (s) at (-5,7);
\coordinate (u) at (-4.93,-1);
\coordinate (v) at (-4.93,-4);
%\coordinate (x) at (-3,0);

\draw[ line width=0.7mm,black,opacity=0.5] (t) node[left][black][scale=0.8] () {$t$} -- (s) node[left][black][scale=0.8] () {$s$};
%\filldraw [black] (s) circle (2pt);
%\filldraw [black] (t) circle (2pt);
%\filldraw [black] (-4,-0.85) circle (4.5pt);
%\filldraw [black] (-5,2) circle (3.5pt);

\draw (-5,2) node [left][scale=0.8] {$e_1$};
\draw (-3.8,-0.68) node [above][scale=0.8] {$e_2$};

\draw  (u) arc(-90:90:2.5) [ line width=0.7mm,blue,opacity=0.5];
\draw  (v) arc(-90:40:3.3) [ line width=0.7mm,teal,opacity=1];
%\draw (-5,2)[red,line width=0.5 mm] pic[rotate = 0] {cross=4pt};
%\draw (-4,-0.85)[red,line width=0.5 mm] pic[rotate = 20] {cross=4pt};
%\draw (x) node [right] {$x$};
\filldraw [black] (-2.39,1.4) circle (4pt);
\draw (-2.39,1.4) node [right][scale=0.8] {$x$};
\draw[ line width=0.7mm,red,opacity=1] (-5,1.5) -- (-5,2.5);
\node[right,black][scale=0.8] at (-5,1.5) {$b$};
\node[right,black][scale=0.8] at (-5,2.5) {$a$};

\filldraw [red] (-5,2.5) circle(4pt);
\filldraw [red] (-5,1.5) circle(4pt);

\draw[ line width=0.7mm,red,opacity=1] (-4,-0.85) .. controls (-3.5,-.6) .. (-2.91,0) ;
\node[below,black][scale=0.8] at (-4,-0.85) {$d$};
\node[right,black][scale=0.8] at (-2.91,0) {$c$};
\filldraw [red] (-4,-0.85) circle(4pt);
\filldraw [red] (-2.91,0) circle(4pt);
\filldraw [black] (s) circle (4pt);
\filldraw [black] (t) circle (4pt);

                \end{tikzpicture}
                \caption{Detour starts at $x$  on secondary path}

        \end{subfigure}
        \caption{Start of the detour}
        \label{fig:detourpoints}
\end{figure}
\begin{definition} (Detour of a path $P$, start and end of a detour)

The detour of the path $P$ is $P \setminus (st \cup st \DIA e_1)$. %Note that the detour is a contiguous subpath path of $P$.
The start of the detour is the first vertex on $P$, after which $P$ deviates from $st \cup st \DIA e_1$. If this vertex is on $st$, then we say that the detour starts on the primary path. If it is on $st \DIA e_1 \setminus st$, then we say that it starts on the secondary path. Similarly, we can define the end of the detour.
\end{definition}

Note that the above definition of detour differs slightly from what is usually used in literature. Specifically, the start of the detour is generally defined as the vertex on the path $P$, after which $P$ deviates from $st$. In our definition, the detour can start on the primary path or the secondary path. Similarly, the detour may end on the primary or the secondary path.

\begin{comment}
Before we describe our main technical contribution, we redefine the notion of the "intact from faults" and a "clean vertex" is slightly different from \cite{DuanR22}.

\begin{definition}
A path $sx$ is considered intact from failures   $F$ if no edges of $F$ lie on it. Similarly, the subtree under a node $x$ in $T_s$, that is, $T_s(x)$, is intact from failures $F$ if no edges of $F$ lie in $T_s(x)$. The vertex $x$ is called $s$-clean from $F$ if $sx$ and $T_s(x)$ is intact from  failures $F$.
\end{definition}

Note the difference between our notion of {\em intact from failures} with that of \cite{DuanR22}. 
\end{comment}

In this paper, we use an approach similar to that used in \cite{Demetrescu2008}. Essentially, our objective is to identify vertices along the path $P = st \DIA \{e_1, e_2\}$ that are `close to the detour' from the vertices $s$ and $t$. Before describing what we exactly mean by `close to the detour', let us define a suitable set of vertices called landmark vertices. These landmark vertices play a crucial role in the forthcoming definitions.

\begin{definition} (Landmark vertices)
\label{def:landmark}
Let $\LL_i$ be the set of vertices sampled from $V$ with a probability of $\frac{c\log n}{2^i}$, where $0 \leq i \leq \log n$ and $c>1$ is a constant. The size of $\LL_i$ is $\TL\left(\frac{n}{2^i}\right)$ with a probability $\ge (1-\frac{1}{n^c})$.
%For brevity, a vertex $u$ in $\LL_i$ is denoted as $u^i$.
\end{definition}
\vspace{-0.2 cm}
We are now ready to describe our main contribution in this paper via the following two definitions:

\begin{definition} ( Close vertex to a faulty edge from $s$)

Let $P = st \DIA \{e_1,e_2\}$. A vertex $x \in P$ is called close vertex  to edge  $e_1$  from $s$ if $x \in \LL_{\EL}$ (for some $\EL \in [0, \log n]$), $sx$ avoids $\{e_1,e_2\}$ and $|xe_1 \DIA \{e_1,e_2\}| = O(2^{\EL})$. The definition is similar for close vertex to $e_2$. Similarly, we can define close vertex to a faulty edge from $t$.

\end{definition}

In most parts of the paper, we will find a close vertex from $s$. So, for brevity, we will drop the term "from $s$" where it is clear from the context that we are finding a close vertex from $s$. We aim to find a close vertex with one more special property. Next, we define the detour close vertex or D-close vertex. This will be close to faulty edge $e_1$ or $e_2$ depending on where the detour starts. 

\begin{definition}(A detour close or D-close vertex to a faulty edge from $s$)

Let $P = st \DIA \{e_1,e_2\}$. Let us assume that the detour of $P$ starts on the primary path.
A vertex $x \in P$ is called D-close to $e_1$ if $x$ is close to $e_1$ and the detour of $P$ starts in $xe_1 \DIA \{e_1,e_2\}$ path ($xe_1 \DIA \{e_1,e_2\}$ represents the path from $x$ to $e_1$ on the primary path). If the detour of $P$ starts on the secondary path, the definition of D-close to $e_2$ is similar (with $xe_1\DIA \{e_1,e_2\}$ replace by $xe_2 \DIA \{e_1,e_2\}$, which represents the path from $x$ to $e_2$ on the secondary path). Similarly, we can define a $D$-close vertex from $t$.
\end{definition}
Our central intuition is that once we have found a $D$-close vertex from $s$ and $t$, we can use  Ren and Duan's approach to find an $s$-clean or $t$-clean vertex. Note that finding a $D$-close vertex is not the main aim. The aim is to find an intermediate vertex. If we find a clean or intermediate vertex at any point in our algorithm, we abandon our search for a $D$-close vertex as our job is already done. We now design our version of \textsc{HitSet} algorithm in the \Cref{alg:query}.  In \Cref{fig:ourapproach}, we show behaviour of our \textsc{HitSet} algorithm in the worst case: the first call may give us a $D$-close vertex from $s$. The second call gives us a $D$-close vertex from $t$. After that, we will find either an $s$-clean vertex or a $t$-clean vertex, and the remaining figure is same as \Cref{fig:ourapproach}.
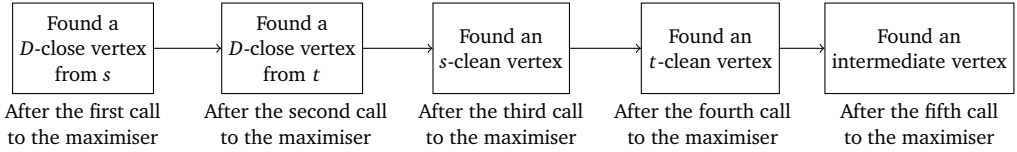
\begin{figure}[hpt!]
\centering
\begin{tikzpicture}[node distance=2.75cm, box/.style={draw, rectangle, minimum width=1cm, minimum height=1.5cm,scale=0.8}]

 % First Box
\draw (0,0)% circle (1.5cm)
    node[align=center,minimum size=3cm,draw,box] (box1)
    {Found a \\$D$-close vertex\\ from $s$};
% Second Box
\draw (0,0)% circle (1.5cm)
    node[right of=box1,align=center,minimum size=3cm,draw,box] (box2)
    {Found a \\$D$-close vertex \\ from $t$};

% Third Box
\draw (0,0)% circle (1.5cm)
    node[right of=box2,align=center,minimum size=3cm,draw,box] (box3)
    {Found an \\$s$-clean vertex};
% Fourth Box
\draw (0,0)% circle (1.5cm)
    node[right of=box3,align=center,minimum size=3cm,draw,box] (box4)
    {Found an \\$t$-clean vertex};
% Fifth Box
\draw (0,0)% circle (1.5cm)
    node[right of=box4,align=center,minimum size=3cm,draw,box] (box5)
    {Found an \\intermediate vertex};

% Arrow
\draw[->] (box1) -- (box2);
\draw[->] (box2) -- (box3);
\draw[->] (box3) -- (box4);
\draw[->] (box4) -- (box5);

\node[align=center,minimum size=0.5cm, scale=0.8] at (0,-1)
    {After the first call\\ to the maximiser};
\node[align=center,minimum size=0.5cm, scale=0.8] at (2.8,-1)
    {After the second call\\ to the maximiser};
\node[align=center,minimum size=0.5cm, scale=0.8] at (5.5,-1)
    {After the third call\\ to the maximiser};
\node[align=center,minimum size=0.5cm, scale=0.8] at (8.15,-1)
    {After the fourth call\\ to the maximiser};
\node[align=center,minimum size=0.5cm, scale=0.8] at (11,-1)
    {After the fifth call\\ to the maximiser};
\end{tikzpicture}
\caption{After at most five consecutive calls to the maximiser, we are sure to find an intermediate vertex.}
\label{fig:ourapproach}
\end{figure}
In essence, our result implies we should first find a vertex on $P$ that is  $D$-close to the fault. Similar kind of strategy can be found in \cite{Demetrescu2008}. This strategy not only decreases the query time but also improves the space taken by the algorithm. Finding a  $D$-close vertex is easy when there is one fault. In our paper, we show how to find a $D$-close vertex (using maximisers) when there are two faults.
Unfortunately, we will not be able to use the maximisers of Ren and Duan\cite{DuanR22} as its size is $O(n^4)$. We will design new maximisers different from \cite{DuanR22}, taking just $\TL(n^2)$ space. In fact, all the maximisers used in our  $\HIT$ algorithm take $\TL(n^2)$ space. 
This completes the description of our overview.
\vspace{-0.4 cm}

\section{Preliminaries}
\label{sec:prelim}

Till now, we have discussed about the shortest path and shortest distance between any two vertices. In our paper, we also use notion of the shortest distance between a vertex and an edge. For an edge $e = (u,v)$ and a vertex $s$, $|se|$ is defined as the minimum of $|su|$ and $|sv|$.  Also, $|se \DIA F| = \min\{|su \DIA F|, |sv \DIA F| \}$.  When considering an arbitrary path $P$, in which vertex $u$ appears before vertex $v$, $P[u,v]$ refers to the subpath from $u$ to $v$, and $|P[u,v]|$ represents the length of this subpath. Also, $V(e)$ denotes the set containing both the endpoints of $e$, i.e., $V(e) = \{u,v\}$. $T_s$ denotes the shortest path tree rooted at $s$. For any vertex $x$ in $T_s$, $T_s(x)$ represents the subtree rooted at node $x$ in $T_s$. In our algorithm, we may need to quickly find the least common ancestor (\LCA) of any pair of vertices $u$ and $v$ in $T_s$. We will utilise the following result to find the LCA efficiently:

\begin{lemma}
\label{lem:lca}(See \cite{BenderF00} and its references) Given a tree $T$ with $n$ vertices, we can construct a data structure of size $O(n)$ in $O(n)$ time, allowing us to answer LCA queries in $O(1)$ time.
\end{lemma}

We will use the following lemma to show the existence of landmark vertex (defined in \Cref{def:landmark}) with high probability on sufficiently long paths. We state this lemma without proof.

\begin{lemma}
\label{lem:highprobability}
Let $U_i$ be the set of paths such that length of each path is $\Omega(2^i)$. If size of $U_i$ is poly($n$), then with a probability $\ge 1-1/n$, for all $0 \le i \le \log n$ and for each path $P \in U_i$, there exists a vertex in $\LL_i$ that hits $P$.
\end{lemma}

We use another concept from the works of \cite{ChechikCFK17} which is named as { \em trapezoid} of a path or the near vertices of a path. Before that, let us describe another notion below: 

\begin{definition}(Ball from a vertex of length $\EL$, $\BA_G(u,\EL)$ )

\noindent A ball from a vertex $u$ of length $\EL$ is a set containing all the vertices at distance $\le \EL$ from $u$. We denote this ball by $\BA_G(u,\EL)$. Formally, $\BA_G(u,\EL) = \{v \ |\ |uv| \leq \EL \}$.
\end{definition}

Let $P=st \DIA F$. Let us now define trapezoid of that path.

\begin{definition}(Trapezoid of $P$ in $G \setminus F$, $\TR_{G \setminus F}(P)$)

\noindent Let $P$ be an arbitrary path from $s$ to $t$ in $G \setminus F$. Define the trapezoid of $P$ in $G \setminus F$ as:

 $$\TR_{G \setminus F}(P)=\Big(\bigcup_{u \in P \setminus \{s,t\}} \BA_{G/F}(u, \epsilon \min\{|P[s,u]|,|P[u,t]|\})\Big)\setminus \{s,t\}$$

 where $\EP \ge 0$ is a parameter. In words, for each vertex $u \in P$, we include all vertices that are in the ball of length $\EP \EL_u$ in $G \setminus F$ where $ \EL_u$ is the minimum of the length of $u$ from $s$ and $t$ on path $P$. All these vertices lie in the trapezoid of $P$. A vertex is  \em{far away}  from $P$ if it does not lie in the trapezoid of $P$.
\end{definition}
Chechik et al.\cite{ChechikCFK17} introduced the notion of the trapezoid and designed a data structure $\FT$. Though we will not use the $\FT$ data structure in our paper, we will crucially use the notion of a trapezoid.
In \cite{DuanR22}, the authors introduced a technique called the maximiser for a pair of vertices $s$ and $t$.
\begin{equation}
\label{eq:1} 
\DD(s,t, x,y, b_1, b_2) = \text{argmax}_{F \in E^2}\left\{ \text{conditions depending on $s$, $t$, $x$, $y$, $b_1$ and $b_2$ }\right\}
\end{equation}

In the above maximiser, $s$ is the source, and $t$ is the destination. Variables $x$ and $y$ can be any two arbitrary vertices. The boolean variables $b_1, b_2$ gave some extra conditions on the pair of faults they are considering. Let $e_1^*$ and $e_2^*$ be the edges returned (or stored) by the above maximiser. Note that we can also store the path length $|st \DIA \{e_1^*,e_2^*\}|$ without increasing the size.

We first describe the power of these maximisers. Let us assume that we want to find the shortest path from $s$ to $t$ avoiding  $\{e_1,e_2\}$. If $\{e_1,e_2\}$ satisfy the conditions of the maximiser $\DD$, then we can show that edges returned by the maximiser $\DD$ hit the path $st \DIA \{e_1,e_2\}$ or we have got the length of the path $st \DIA \{e_1,e_2\}$.

\begin{lemma}
\label{lem:notinpathP}
        Let $P=st \diamond \{e_1,e_2\}$ and $\DD(s,t,x,y) = \{e_1^*,e_2^*\}$.  If $\{e_1,e_2\}$ satisfies conditions of $\DD$, then  either $e_1^*$ or $e_2^*$  lie on $P$ or $|P|=|st \DIA \{e_1^*,e_2^*\}|$.
\end{lemma}
\begin{proof}
Let us assume that neither $e_1^*$ nor $e_2^*$ lies on $P$.
    Since $\{e_1,e_2\}$ also satisfies the condition given in $\DD$,  $|P| \leq |st \DIA \{e_1^*,e_2^*\}|$.

        Since $P$ also avoids $e_1^*$ and $e_2^*$,  $|st \DIA \{e_1^*,e_2^*\}| \leq |P|$. This implies $|P| = |st \DIA \{e_1^*,e_2^*\}|$.
\end{proof}

Using the above lemma, either we find the length of the path $st \DIA \{e_1,e_2\}$ or we get a very small set that hits $P$. We now formally describe the maximisers created in \cite{DuanR22} and how they use the notions of $s$-clean and $t$-clean vertices. Let $bits$ be an array of size 2 (we replace $b_1,b_2$ used by them by $bits$ for our presentation). Each cell in bits can either be 0 or 1. We will now define four maximisers $\DD_{bits}(s,t,x,y)$ depending on the value in $bits$. The conditions of the maximisers are as follows:(1) $sx$ and $ty$ are intact from faults and (2)  If $bits[0] = 1$, then  $x$ is $s$-clean. Similarly, if $bits[1]=1$, then  $y$ is $t$-clean. We give the complete definition of one of the maximisers, $\DD_{01}(s,t, x,y) = \text{argmax}_{\{e_1^*,e_2^*\} \in E^2}\left\{ \text{$sx$ is intact from  faults $\{e_1^*,e_2^*$\} and $y$ is $t$-clean from $\{e_1^*,e_2^*\}$ } \right\}$.

In the above definition, $s,t,x,y \in V$. Thus, the total size of $\DD_{01}$ is $O(n^4)$. It is easy to observe that all the other three maximisers will also be of size $O(n^4)$.  We have already discussed in \Cref{sec:overview} about how these maximisers are used. In this paper, we have modified these maximisers to have an oracle with size $\TL(n^2)$.

\section{Oracle for two faults}
\label{sec:twofault}

%Let us again reiterate the setting we will use throughout the paper. We want to find the shortest path from $s$ to $t$ avoiding $\{e_1,e_2\}$. Let $P = st \DIA \{e_1,e_2\}$ such that $P$ is $2$-decomposable. Also, $e_1$ lies on the primary path $st$ and $e_2$ lies on the secondary path $st \DIA e_1$.

Let us reiterate our setting:  we want to find the shortest path from $s$ to $t$ avoiding $e_1=(a,b)$ and $e_2=(c,d)$. We will assume that $e_1 \in st$ and $e_2 \in st\DIA e_1$ where $a$ and $c$ are near $s$ on $st$ and $st \DIA e_1$ paths, respectively. Let $P = st\DIA \{e_1,e_2\}$ be a 2-decomposable path.

First, we describe some basic data structures used in our algorithm. As mentioned in \Cref{sec:prelim}, we will add small weights to each edge to ensure all the shortest paths are unique. Let $G'$ be the graph obtained after adding weights. In $G'$, we will find the shortest path tree from each vertex. The space taken by all shortest-path trees is $O(n^2)$. Similarly, whenever we need to find $st \DIA e_1$, we will apply the algorithm of \cite{Demetrescu2008} on $G'$. Again the algorithm in \cite{Demetrescu2008} takes $O(n^2)$ space. Apart from this data structure, we will never use $G'$ in any of our other data structures. For the rest of the paper, our graph is unweighted.

Our new data structure is the maximiser $\DD_{bits}(s,t,dist,clean)$. It stores the pair of edges $(e_1^*,e_2^*)$ that maximises the distance from $s$ to $t$ subject to conditions determined by various parameters. We will also store the length of $st \DIA \{e_1^*,e_2^*\}$. Formally,

\begin{equation}
\label{eq:maximiser}
\DD_{bits}(s,t,dist,clean) = \arg\max_{F \in E^2}\{ \text{conditions that depend on $bits, s, t, dist$, and $clean$} \}
\end{equation}

One of the common conditions in all our maximisers in this section is as follows:
\begin{align*}
\text{Common Condition: A pair of edge $(e,e')$ will be considered in maximiser only if $e \in st$, $e' \in st \DIA e$.}  
\numberthis \label{eq:common}
\end{align*}
Note that this condition is satisfied by the pair $(e_1,e_2)$. We now describe other parameters of the maximiser.
$bits$, $dist$, and $clean$ are sets of size two implemented using arrays. Each element of $bits$ can either be 0,1, or 2. If $bits[0] = 0$, then we have not found  a $D$-close or $s$-clean vertex from $s$. If $bits[0]=1$, then we have found a $D$-close vertex from $s$. Moreover, if $bits[0] =2$, then we have found an $s$-clean vertex. We store the corresponding $s$-clean vertex in $clean[0]$. The same conditions hold for $bits[1]$ and $clean[1]$, which tells us whether we have found a $D$-close vertex from $t$ or a $t$-clean vertex and stores the $t$-clean vertex if found. $dist[0]$ and $dist[1]$ normally store the distances from $s$ and $t$ respectively. There are nine versions of our maximiser $\DD$, which depend on the value of $bits$. We will define them on the fly, as and when our algorithm needs them. We will discuss the space taken by these maximisers and our oracle in \Cref{sec:space}.
Before we describe our algorithm, let us take care of some corner cases.

We will always use maximiser so that $\{e_1,e_2\}$ satisfies its conditions. Let $\{e_1^*,e_2^*\}$ be the edges returned by the maximiser. If both $e_1^*$ and $e_2^*$ do not lie on $P$, then using \Cref{lem:notinpathP}, $|P| = |st \DIA \{e_1^*,e_2^*\}|$. Thus, we can easily find $|P|$ in this case. In the worst case, $e_1^*$ and/or $e_2^*$ may hit $P$. We will describe our algorithm for the worst case henceforth.

In our algorithm, let us assume that we have found a vertex $x$ which is one of the endpoints of $e_1^*$ or $e_2^*$ and also lies on the path $P$. First, we check if $x$ lies in the first or last segment of $P$ or $x$ is an intermediate vertex. For this, we need to check if $sx$ (and $tx$) is intact from faults $\{e_1,e_2\}$. Using \Cref{lem:lca}, this can be done in $O(1)$ time using appropriate $\LCA$ query in $T_s$  (and $T_t$). Similarly, we can check if $T_s(x)$ contain faults $\{e_1,e_2\}$ in $O(1)$ time. Thus, processing $x$ takes $O(1)$ time.

 If $x$ happens to be  $s$-clean, we will set $clean[0] = x$ in the subsequent calls to the maximiser. If $x$ is $D$-close from $s$, then we can recurse our algorithm on the path $P[x,t]$ on which $x$ is acting as the source. Formally, if $x$ is $D$-close from $s$, our answer is $|sx| + \QU(x,t,\{e_1,e_2\},2)$.
It may be that the primary path $xt$ contains both the faults $e_1$ and $e_2$. In that case,  we will use our algorithm in \Cref{both} to find $|P[x,t]|$. Similarly, if the secondary path from $x$ to $t$ does not contain any faulty edge, then we can find $|xt \DIA \{e_1,e_2\}|$ using the single fault algorithm of \cite{Demetrescu2008}. Again, these are easy cases for us. Henceforth, we assume that for any $D$-close vertex $x$, the primary and the secondary path from $x$ to $t$ contain faulty edges.

Finally, we will describe our approach to finding a $D$-close vertex from $s$, an $s$-clean vertex, and the intermediate vertex. Approach to finding these vertices from the destination side $t$ is symmetrical.

\section{Finding a $D$-close vertex from $s$}
\label{sec:findDclose}

By definition, a $D$-close vertex can be on the primary path or the secondary path depending on where the detour of $P$ starts. %Thus, there are two cases: (1) detour of $P$ starts on the primary path and (2) detour of $P$ starts on the secondary path.
We deal these cases seperately.

\subsection{The detour  starts on the primary path}
\label{sec:startprimary}
We first describe the maximiser. Since we have still not found a $D$-close vertex from $s$, $bits[0]=0$. Our maximiser is $\DD_{0\alpha}(s,t,dist,clean)$, where $\alpha $ can be $\{0,1,2\}$. We  set $dist[0] = O(2^{\lfloor\log |se_1|\rfloor})$. 

The conditions of the maximiser are as follows:

\begin{align*}
\DD_{0\alpha}(s,t,dist,clean) =  \arg\max_{F \in E^2}\Bigl\{& dist[0]  \text{ distance from $s$ on the primary path is intact from $F$ and other}\\ & \text{ conditions based on $\alpha$ from $t$ side.}\Bigl\}
 \numberthis \label{eqn:3}
\end{align*}
\begin{figure}
\centering
\begin{tikzpicture}[scale=0.7]
\coordinate (s) at (-6,8);
\coordinate (a) at (-7,4);
\coordinate (x) at (-5,6);
\coordinate (y) at (-5,2);
\coordinate (z) at (-2,2);
\coordinate (c') at (-3.8,4);
\coordinate (d) at (-4.3,3.3);
\coordinate (b) at (-7.5,3.5);
\coordinate (c) at (-5,3.8);
\coordinate (t) at (-6.3,2.7);

\draw[ line width=0.7mm,black,opacity=1] (s) .. controls (-8,6) and (-4,6) .. (a);
\filldraw [black] (s) circle (3pt);
\filldraw [red] (a) circle (3pt);
\draw[ line width=0.7mm,black,opacity=1] (s) .. controls (-8,6) and (-3,8) .. (x);
\filldraw [black] (x) circle (3pt);
\draw[ line width=0.7mm,black,opacity=1] (x) .. controls (-5,4) and (-4,3) .. (y);
\draw[ line width=0.7mm,black,opacity=1] (x) .. controls (0,4) and (-2,3) .. (z);
\draw[ line width=0.7mm,black,opacity=1] (x) .. controls (-5,4) and (-4,3) .. (y);
\draw[ line width=0.7mm,black,opacity=1] (y) .. controls (-4,2.5) and (-2,2.5) .. (z);

\filldraw [red] (d) circle (3pt);
\filldraw [red] (b) circle (3pt);

\draw[ line width=0.7mm,red,opacity=1] (c') -- (d);
\draw[ line width=0.7mm,red,opacity=1] (a) -- (b);
\draw[ densely dashed, line width=0.7mm,blue,opacity=1] (x) .. controls (-4,5) and (-4,5) .. (c');
\filldraw [red] (c') circle (3pt);
\draw (x) node[left][scale=0.8] {$x$};
\draw (c') node[right][scale=0.8] {$c'$};
\draw (-4.3,3.5) node[above][scale=0.8] {$d$};
\draw (s) node[above][scale=0.8] {$s$};
\draw (-7,4.2) node[above][scale=0.8] {$a$};
\draw (b) node[below][scale=0.8] {$b$};
\draw[ densely dashed, line width=0.7mm,magenta,opacity=1] (x) .. controls (-5,5) and (-6,5) .. (c);
\filldraw [red] (c) circle (3pt);
\draw (c) node[left][scale=0.8] {$c$};
\draw[ line width=0.7mm,red,opacity=1] (c) -- (d);

\draw[ densely dashed, line width=0.7mm,magenta,opacity=1] (d) .. controls (-3.8,2.7) and (-3.8,2.7) .. (t);

\filldraw [red] (d) circle (3pt);
\filldraw [black] (x) circle (3pt);
\filldraw [black] (t) circle (3pt);
\draw (-6.3,3,2) node[above][scale=0.8] {$t$};

\end{tikzpicture}
\caption{If $d \in T_s(x)$, the paths $sd \diamond e_1$ is not a prefix of  $st \diamond e_1$}
\label{fig:dnotinTsx}
\end{figure}
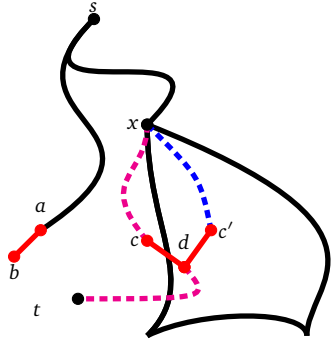

 Let $x $ be an endpoint of an edge returned by the maximiser. If $x$ is an intermediate vertex, we abandon our search for $D$-close vertex. So, let us assume that $x$ lies on the first segment of $P$. We will show that $x$ lies on the $st$ path or $x$ is $s$-clean.

  \begin{lemma}
  \label{claim:xinst}
        Let $x$ be an endpoint of an edge returned by the maximiser (\ref{eqn:3}) such that $sx$ is intact from failures $\{e_1,e_2\}$. If $x$ does not lie on the $st$ path, then $x$ is an $s$-clean vertex.
  \end{lemma}

  \begin{proof}
        Let us assume that $x \notin st$. We will now show that $x$ is $s$-clean. Since $sx$ is intact from failures $\{e_1,e_2\}$, we only need to show that endpoints of $e_1$ and $e_2$, i.e., $a,b,c,d$ do not lie in $T_s(x) $.

        Since $x$ is not in the $st$ path and hence not in the $sa$ path,  the two subtrees $T_s(x)$ and $T_s(a)$ are disjoint. Hence, $a,b \notin T_s(x)$.

    Now, if $c \in T_s(x)$, then the path $sc$ is intact from failures and is a subpath of the secondary path (as we assumed that the secondary path passes through $e_2$ and hence $c$). Also, $x$ is in the path $P$. So, in this case, the detour will start in the subpath $xe_2$  on the secondary path. This is a contradiction.

If $d \in T_s(x)$, then since $c \notin T_s(x)$, there exists some $c'$ such that $c'd$ is an edge of $T_s(x)$. Consider the path from $s$ to $x$ concatenated with the path from $x$ to $d$ in $T_s(x)$ (See \Cref{fig:dnotinTsx}). This concatenated path represents the shortest path from $s$ to $d$ and avoids $e_1$. Additionally, it should be noted that $T_s$ was constructed by adding small weights to each edge. Therefore,  $|sd \DIA  e_1| = |sd|$ is unique and does not pass through $e_2$. Furthermore, it should be remembered that we found $st \DIA  e_1$ in the same weighted graph, and thus $st \DIA e_1$ is also unique and passes through $e_2$ (we assumed at the start of \Cref{sec:twofault} that $e_2 \in st \diamond e_1$). However, according to the property of shortest paths, $sd \DIA e_1$ should be a subpath of $st \DIA e_1$. This leads to a contradiction, implying that $d \notin T_s(x)$.   \end{proof}

 \begin{figure}
        \centering
                \begin{tikzpicture}[scale=0.27]
                \coordinate (t) at (-5,-5);
\coordinate (s) at (-5,11);
\coordinate (u) at (-4.93,-1);
\coordinate (v) at (-4.93,-4);
%\coordinate (x) at (-3,0);

\draw[ line width=0.7mm,black,opacity=0.5] (t) node[below][black] () {$t$} -- (s) node[above][black] () {$s$};
%\filldraw [black][opacity=0.5] (s) circle (2pt);
%\filldraw [red] (t) circle (2pt);

%\filldraw [red] (-5,2) circle (3.5pt);

\draw (-5,2) node [left][scale=0.8] {$e_1$};
\draw (-3.8,-0.68) node [above] [scale=0.8]{$e_2$};

\draw  (u) arc(-90:90:2.5) [ line width=0.7mm,blue,opacity=0.5];
\draw  (v) arc(-90:90:4.5) [ line width=0.7mm, teal ,opacity=1];
%\filldraw [red] (-4,-0.85) circle (4.5pt);
%\draw (-4,-0.85)[red,line width=0.5 mm] pic[rotate = 20] {cross=4pt};
%\draw (-5,2)[red,line width=0.5 mm] pic[rotate = 0] {cross=4pt};
\filldraw [black] (-5,6) circle(4pt);
\draw (-5,6) node [right][scale=0.8] {$x$};
\filldraw [black] (-5,8) circle(4pt);
\draw (-5,8) node [right][scale=0.8] {$y \in \LL_l$};

\draw[ line width=0.7mm,red,opacity=1] (-5,1.5) node[right][black][scale=0.8] () {$b$} -- (-5,2.5) node[right][black][scale=0.8] () {$a$};
\filldraw [red] (-5,2.5) circle(4pt);
\filldraw [red] (-5,1.5) circle(4pt);

\draw[ line width=0.7mm,red,opacity=1] (-4,-0.85) node[below][black][scale=0.8] () {$d$} .. controls (-3.5,-.6) .. (-2.91,0) node[right][black][scale=0.8] () {$c$};
\filldraw [red] (-4,-0.85) circle(4pt);
\filldraw [red] (-2.91,0) circle(4pt);

\draw [decorate,
    decoration = {calligraphic brace}][ultra thick] (-5.4,6.2) --  (-5.4,10.9);
\node[draw][scale=0.8] at (-8,8.5) {$\geq 2^l$};
    
  \draw [decorate,
    decoration = {calligraphic brace}][ultra thick] (-5.4,2.5) --  (-5.4,5.8);
 \node[draw][scale=0.8] at (-8,4.2) {$O(2^l)$};
 \filldraw [black] (s) circle (4pt);
\filldraw [black] (t) circle (4pt);

                \end{tikzpicture}
                \caption{Finding a $D$-close vertex}
        \end{figure}
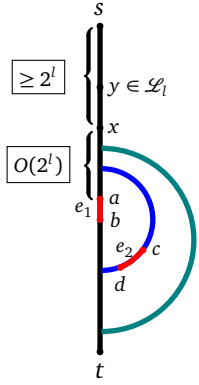
   If $x$ is not an $s$-clean vertex, then using \Cref{claim:xinst}, $x$ lies on $st$ path. We will now show how to find a  $D$-close vertex from $s$ using $x$. Let us assume that $\EL = \lfloor\log |se_1|\rfloor$. By the first condition of the maximiser (\ref{eqn:3}), $|sx| \ge 2^{\EL}$. Also, since $|se_1| = O(2^{\EL})$,   $|xe_1| = O(2^{\EL})$.  We now find the first vertex, say $y$, from $x$ on the $xs$ path that lies in $\LL_{\EL}$. Using \Cref{lem:highprobability}, such a vertex exists with a high probability within a distance $2^{\EL}$ from $x$. By construction, the distance $|ye_1|$ on the primary path is also $ O(2^{\EL})$. This implies that $y$ is {\em close} to $e_1$. We now show that $y$ is $D$-close  to $e_1$ from $s$.

To this end, remember that $sx$ is intact from failures and $x$ lies on $P$ as well as $st$ path. Thus, by construction, even $sy$ is intact from failures and $y$ lies in $P$  as well as $st$ path. This implies that the detour of $P$ starts on the primary path $st$ after the vertex $y$ or on the subpath $ye_1$ of $st$. Thus, $y$ is $D$-close to $e_1$ from $s$.

\subsection{The detour starts on the secondary path}
\label{sec:startsecondary}
We aim to find a $D$-close vertex to the edge $e_2$. The reader will see that this is a challenging case. In this case, we will use a different maximiser. We set $dist[0] = O(2^{\lfloor\log \min\{|se_1|,|se_2 \DIA e_1|\}\rfloor})$.  Note that $se_2 \DIA e_1$ will not pass through $e_2$. Thus, the single fault data structure of \cite{Demetrescu2008} can find it in $O(1)$ time. The conditions in the maximiser are as follows:

\begin{align*}
\DD_{0\alpha}(s,t,dist,clean) =  \arg\max_{F \in E^2}\big\{& dist[0]  \text{ distance from $s$ on the primary path and \textbf{secondary} path is intact} \\ & \text{ from $F$ , other conditions based on $\alpha$ from $t$ side.} \big\} \numberthis \label{eqn:4}
\end{align*}

 Some remarks are in order. Specifically, we want the reader to note the use of primary and secondary paths in the maximiser. Let $(e,e')$ be the pair of edges that satisfy the common condition of \Cref{eq:common}. For this pair, the primary path is $st$  (it always remains the same), and the secondary path is $st \DIA e$ (the secondary path changes for each pair). We demand that $dist[0]$ from both the primary and the secondary paths are intact from failures. The reader can check that the pair $e_1,e_2$ satisfies the conditions of the maximiser.

Let  $x$ be the endpoint of an edge returned by the maximiser such that it lies in the first segment of $P$. Let us assume that $x$ is not $s$-clean. We will now show how to find a $D$-close vertex using $x$.
There are two cases here. Let us look at both of them, starting with the easier case.

(1) $|se_2 \DIA e_1| \le |se_1|$ : This case is similar to the one we discussed in the previous section. Let $\EL =\lfloor\log |se_2 \DIA e_1|\rfloor$. By the conditions of the maximiser, $|sx| \ge 2^{\EL}$. As in the previous section, we will find a vertex $y$ in $xs$ path that lies in $\LL_{\EL}$. Again, the reader can verify that $y$ is a $D$-close vertex to $e_2$ from $s$.

(2) $|se_1| < |se_2 \DIA e_1|$ :This is the most complex case for us. Unlike the previous case, we cannot find a $D$-close vertex using one invocation of the maximiser. The reader will see that we will use different maximisers two more times. We will show that vertices returned by these maximisers will slowly lead us to the desired $D$-close vertex.

\begin{enumerate}
\item  The first maximiser (stated above) will give us a vertex $y$, which is {\bf close} to $e_1$. Note that we want to find a $D$-close vertex to $e_2$.

\item The second maximiser will give us a vertex $p \in P$ that has one important property: $a$ does not lie in the trapezoid of $P[p,t]$. This is the only place in our paper where we use the concept of the trapezoid.

\item Once we have $p$, we can find the $D$-close vertex to $e_2$ in one call to an appropriate maximiser. This $D$-close vertex, say $y$, also satisfies the above property. That is, $a$ does not lie in the trapezoid of $P[y,t]$.
\end{enumerate}
Before we move further, we describe another issue. As stated above, in this case, we will find two vertices, $y$ and $p$. After we found, say, $y$, it may be the case that the primary path $yt$ passes through $e_2$ and the secondary path passes through $e_1$. This is actually a good case for us as we recurse our operations in the  \Cref{sec:startprimary} to find a $D$-close vertex to $e_2$ from $y$. In the rest of the  section, we will assume the following:

\begin{assumption}
\label{as:primarycontainse1}
All primary paths contain $e_1$, and the secondary paths contain $e_2$.
\end{assumption}

Together with this assumption, we also show that, while recursing on $P$ having any vertex as the source, if $e_1 = (a,b)$ lies on the primary path, then $a$ is closer to the source. Similarly, if $e_2 = (c,d)$ lies on secondary path, then $c$ is closer to the source. Remember that we had assumed this for the source $s$ at the start of this section. This section will show that this assumption remains valid for all primary and secondary paths.

\begin{lemma}
        Let $e_1=(a,b)$ be the edge on the primary path $st$ where $a$ is closer to $s$. Similarly, let $e_2 = (c,d)$ lies on $st \DIA e_1$ and $c$ is close to $s$ in $st \DIA e_1$. Also, let $p$ be a vertex on path $P[s,t]$. If $e_1$ lies on the primary path $pt$, then $a$ is closer to $p$. Similarly, if $e_2$ lies on $pt \DIA e_1$, then $c$ is closer to $p$ on $st \DIA e_1$.
\end{lemma}

\begin{proof}
         Since $a$ is closer to $s$ in $st$ path,  $b$ is closer to $t$ than $a$. If $b$ is closer to $p$ on the $pt$ path, then we can directly use the $bt$ subpath of the $st$ path and get the shortest $pt$ primary path without using the edge $e_1$. This violates the condition of the claim.
        Hence, $a$ is closer to $p$.

  The proof for the second part of the claim is similar.
\end{proof}

With this, we have addressed all the small issues. Let us now look at our three maximisers.

\textbf{First maximiser: }Let us first see how to find a close vertex to $e_1$ (See \Cref{fig:sect52}). Let $x$ be a vertex returned by the maximiser such that $x$ lies on the first segment of  $P$.

Let  $\EL =\lfloor\log |se_1|\rfloor$. By the first condition of the maximiser,
$|sx| \ge 2^{\EL}$. We now find the first vertex of $ \LL_{\EL}$ in $sx$ path from $s$ whose distance is at most  $2^{\EL}$. Using \Cref{lem:highprobability}, such a vertex, say $y \in \LL_{\EL}$, exists with a high probability. By \Cref{as:primarycontainse1}, the primary path $yt$ passes through $e_1$. We will now show that $|ye_1| = O(2^{\EL})$. To this end, we  will  find a path from $y$ to $e_1$ that avoids $\{e_1,e_2\}$ and has length $ O(2^{\EL})$. Consider the subpath $ys$ of $sx$ concatenated with $se_1$. Since $sx$ is intact from failures $\{e_1,e_2\}$, so is $sy$. And by construction,  $|sy| = O(2^{\EL})$. Similarly, $se_1$ avoids $\{e_1,e_2\}$ and is also of length $O(2^{\EL})$.  Thus,  $|ye_1| = O(2^{\EL})$. By definition, $y$ is close to $e_1$. 

\textbf{Second maximiser:}
We will now show how to use $y$ to find the vertex $p$. To this end, we define another maximiser.
Let $dist[0] = |ye_1|$. The conditions of the maximiser are as follows:

\begin{align*}
\DD_{0\alpha}(y,t,dist,clean) =  \arg\max_{F \in E^2}\big\{& dist[0]  \text{ distance from $s$ on the primary path and \textbf{secondary} path is intact }\\ & \text{from failures $F$  and other conditions based on $\alpha$ from the $t$ side.}\big\}
 \numberthis \label{eqn:5}
\end{align*}
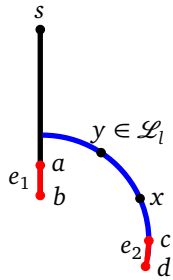
\begin{figure}[hpt!]
                \centering
                \begin{tikzpicture}[scale=0.4]

                \coordinate (t) at (-5,-7);
\coordinate (s) at (-5,7);
\coordinate (u) at (-4.93,-1);
\coordinate (v) at (-4.93,-4);
%\coordinate (x) at (-3,0);

\draw[ line width=0.7mm,black,opacity=0.5] (-5,1.5) node[below][black] () {} -- (s) node[above][black] () {$s$};
%\filldraw [black] (s) circle (2pt);
%\filldraw [black] (t) circle (2pt);
%\filldraw [black] (-4,-0.85) circle (4.5pt);
%\filldraw [black] (-5,2) circle (3.5pt);

\draw (-5,2) node [left] {$e_1$};
\draw (-2,-0.9) node [above] {$e_2$};

\draw  (-4.93,3.5) arc(90:-10:3.5) [ line width=0.7mm,blue,opacity=0.5];
%\draw  (v) arc(-90:40:3.3) [ line width=0.7mm,teal,opacity=1];
%\draw (-5,2)[red,line width=0.5 mm] pic[rotate = 0] {cross=4pt};
%\draw (-4,-0.85)[red,line width=0.5 mm] pic[rotate = 20] {cross=4pt};
%\draw (x) node [right] {$x$};
\filldraw [black] (-1.72,1.4) circle (4pt);
\draw (-1.72,1.4) node [right] {$x$};
\draw[ line width=0.7mm,red,opacity=1] (-5,1.5) node[right][black] () {$b$} -- (-5,2.5) node[right][black] () {$a$};
\filldraw [red] (-5,2.5) circle(4pt);
\filldraw [red] (-5,1.5) circle(4pt);

\draw[ line width=0.7mm,red,opacity=1] (-1.54,-0.85) node[right][black] () {$d$} .. controls (-1.48,-.65) .. (-1.43,0) node[right][black] () {$c$};
\filldraw [red] (-1.54,-0.85) circle(4pt);
\filldraw [red] (-1.43,0) circle(4pt);
\filldraw [black] (-3,2.92) circle (4pt);
\draw[scale=0.3] (-7,9.7) node [above] {$y \in \LL_l$};
\filldraw [black] (s) circle (4pt);

                \end{tikzpicture}
                \caption{$y$ is close to $e_1$}
                \label{fig:sect52}

        \end{figure}
Let $p$ be the endpoint of an edge returned by the maximiser such that $p$ lies on the first segment of $P[y,t]$.
Before we move further, let us take a slight detour in our discussion. Till now, we have seen that if a vertex returned by the maximiser does not lie on the first segment, then it may be used as a clean vertex (for example, maximiser \ref{eqn:3} and \ref{eqn:4}). However, we will not use $p$ as a clean vertex, even if it satisfies the definition of the $s$-clean vertex. This is done to conserve space. Since this is an important technical point, we add this as a remark which will be used crucially when we calculate the space taken by the maximiser.

\begin{remark}
\label{rem:notusecondmaximiser}
Any vertex $p$ returned as the output of the maximiser (\ref{eqn:5}) will not be used as a clean vertex in the subsequent call to another maximiser.
\end{remark}
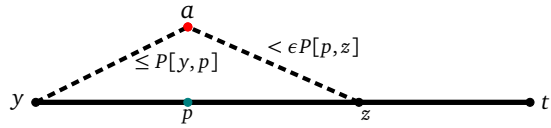
\begin{figure}[hpt!]

\centering
\begin{tikzpicture}[scale=0.5]
\coordinate (y) at (-6,8);
\coordinate (t) at (7,8);
\coordinate (a) at (-2,10);
\coordinate (z) at (2.5,8);

\draw[ line width=0.7mm,black,opacity=0.5] (y) node[left][black][scale=0.9] () {$y$} -- (t) node[right][black][scale=0.9] () {$t$};
\filldraw [red] (-2,10) circle (3pt);
\draw (-2,10) node[above] {$a$};

\draw[densely dashed, ultra thick] (y) -- (a);
\filldraw [black] (z) circle (3pt);
\draw[densely dashed, ultra thick] (z) -- (a);
\filldraw [red] (-2,10) circle (3pt);
\draw (2.7,8) node[below][scale=0.8] {$z$};
\draw (-2,8) node[below][scale=0.8] {$p$}; 
\draw (-2.3,8.6) node[above][scale=0.8] {$\leq P[y,p]$}; 
\draw (1.3,9) node[above][scale=0.8] {$< \epsilon P[p,z]$};

%\filldraw [black] (-2.5,8) circle (3pt);
%\filldraw [black] (-0.5,9.6) circle (3pt);
%\draw (-2.5,8) node [below] {$e_1$};
%\draw (-0.65,9.6) node [left] {$e_2$};

%\draw  (u) arc(0:180:2.5) [ line width=0.7mm,blue,opacity=0.5];
%\draw  (v) arc(0:126:2.5) [ line width=0.7mm,teal,opacity=1];

%\draw (-2.5,8)[red,line width=0.5 mm] pic[rotate = 0] {cross=4pt};
%\draw (-0.55,9.6)[red,line width=0.5 mm] pic[rotate = 30] {cross=4pt};
%\draw (u) node [left] {$u$};
\filldraw [black] (y) circle (3pt);
\filldraw [teal] (-2,8) circle (3pt);
\filldraw [black] (t) circle (3pt);

\end{tikzpicture}
\caption{$a$ does not lie in the trapezoid of $P[p,t]$ }
\label{fig:lemma51}
\end{figure}
The above remark implies that once we have found $p$ on the first segment of path $P$, we will return $|sp| + \QU(p,t,\{e_1,e_2\},2)$. Thus, we will recurse our algorithm with $p$ as the source. But $p$ has one special property that $s$ did not have: the trapezoid of the path $P[p,t]$ ( that is $pt \DIA \{e_1,e_2\}$),  does not contain $a$ --
we now prove this claim.

\begin{lemma}
\label{claim:trap}
        Let $p$ be a vertex returned by maximiser (\ref{eqn:5}) such that $p$ lies on the first segment of $P[y,t]$. Then, $a$ does not lie in the trapezoid of the path $P[p,t]$. Or by definition of trapezoid, for all $z \in P[p,t]$, $|za \diamond \{e_1,e_2\}| > \epsilon \min\{|P[p,z]|,|P[z,t]|\}.$
\end{lemma}

\begin{proof}
        For contradiction, let $a$ lie in the trapezoid of the subpath $P[p,t]$. Then, there exists a vertex $z$ on $P[p,t]$ such that $|za \DIA \{e_1,e_2\}| \leq \epsilon \min\{|P[p,z]|,|P[z,t]|\} \le  \epsilon |P[p,z]|$.
        Now we will find a path from $y$ to $z$ that avoids all the faults and is less than $P[y,z]$. Consider the path $ya$ concatenated with $az \DIA \{e_1,e_2\}$ (See \Cref{fig:lemma51}). By \Cref{as:primarycontainse1}, $ya$ does not contain any faults. Thus, the above path avoids both $e_1$ and $e_2$.

Let us now calculate the length of this path.
\begin{tabbing}
$|ya| + |az \DIA \{e_1,e_2\}|$ \= = $|ya| + |za \DIA \{e_1,e_2\}|$ \\
    \>  $ \le |yp| + |za \DIA$\=$ \{e_1,e_2\}|$ \hspace{0.25 cm} (By definition of the maximiser, $|yp| \geq |ya|$)\\
 \> $\le |P[y,p]| + |za \DIA$\=$ \{e_1,e_2\}|$ \hspace{0.8 cm}
(as $P[y,p] \ge |yp|)$\\
 \> $\leq |P[y,p]| + \EP |P[p,z]|$ $< |P[y,p]| + |P[p,z]|$ $ = |P[y,z]|$
\end{tabbing}

This is a contradiction as $P[y,z]$ is the shortest path from $y$ to $z$ avoiding both faults.
        So, our assumption was wrong. So, for any $z \in P[p,t]$, $|za \DIA \{e_1,e_2\}| > \epsilon |P[p,z]|$. In other words, $a$ lies outside trapezoid of $P[p,z]$.
\end{proof}

%Thus, we have  found a vertex $p$ such that the primary path $pt$  passes %through $e_1$ and secondary path $pt \DIA e_1$ passes through $e_2$ (by %\Cref{as:primarycontainse1}), but $a$ does not lie in the trapezoid of $P[p,t]$.
We will now use the vertex $p$ to find a $D$-close vertex to $e_2$. To this end, we need another useful claim to design the maximiser.
\begin{lemma}
\label{claim:LCA}
        Let $z$ be the least common ancestor of $e_1$ and $e_2$ in the path $pe_1$ and $pe_2 \DIA e_1$. Let $i$ be the greatest integer such that $(1+\EP)^i$ is less than $|pa|$ on the primary path $pt$. Then the length of the path $|pz|$ on the primary path is $\le (1+\epsilon)^{i}$. See \Cref{fig:lca} for an illustration of the lemma.
\end{lemma}

\begin{proof}
        Assume for contradiction that $|pz| > (1+\epsilon)^{i}$.  Then,

\begin{tabbing}
$|za \DIA \{e_1,e_2\}|$ \=   $ = |za|$ $=|pa|-|pz|$ $ \leq (1+\epsilon)^{(i+1)}-(1+\epsilon)^{i}$ $=\epsilon (1+\epsilon)^{i}$ $ \leq \epsilon |pz|$ $ = \EP|P[p,z]|$
\end{tabbing}

   This violates  \Cref{claim:trap}. Hence, our assumption was wrong.
\end{proof}

\textbf{Last maximiser:} We are now ready to define our last maximiser. First we will find the largest $i$ such that $(1+\EP)^i \le |pe_1| $. Also, we set $dist[0] = 2^{\lfloor\log |pe_2 \DIA e_1|\rfloor}$. We will use the following conditions in the maximiser:

\begin{align*}
& \DD_{0\alpha}(p,t,dist,clean) =  \arg\max_{F \in E^2}\{   \text{ $(1+\EP)^i$ distance from $p$ on the primary path is intact from $F$ ; }\\ & \text{ \hspace{3 cm} $dist[0]$ distance from $p$ on the secondary path is intact from $F$; Other } \\  & \text{ \hspace{3 cm} conditions based on $\alpha$ from $t$ side.}\}
 \numberthis \label{eqn:6}
\end{align*}

\begin{figure}[hpt!]
\centering
\begin{tikzpicture}[scale=0.6]
\coordinate (p) at (-4,8);
\coordinate (a) at (-7,4);
\coordinate (c) at (-3,4);
\coordinate (d) at (-2.5,3.3);
\coordinate (b) at (-7,3);
\coordinate (z) at (-5.8,7);

\draw[ line width=0.7mm,black,opacity=1] (p) ..node[above, rotate=70, scale=0.7] {\(pe_1\)} controls (-8,6) and (-7,6) .. (a);
\draw[ line width=0.7mm,black,opacity=1] (p)  ..node[pos=0.8,above, rotate=-30, scale=0.7] {\(pe_2 \DIA e_1\)} controls (-8,6) and (-6,6) .. (c);
\filldraw [black] (p) circle (3pt);
\filldraw [black] (z) circle (3.5pt);
\filldraw [red] (a) circle (3pt);
\draw[ line width=0.7mm,black,opacity=1] (-7.2,4.6) -- (-6.85,4.6);

\filldraw [red] (c) circle (3pt);
\filldraw [red] (d) circle (3pt);
\filldraw [red] (b) circle (3pt);

\draw[ line width=0.7mm,red,opacity=1] (c) -- (d);
\draw[ line width=0.7mm,red,opacity=1] (a) -- (b);
%\draw[ densely dashed, line width=0.7mm,blue,opacity=1] (x) .. controls (-3,5) and (-5,5) .. (c);
\filldraw [red] (c) circle (3pt);

%\draw (c) node[right][scale=1.2] {$c$};
\draw (d) node[right][scale=0.8] {$d$};
\draw (c) node[right][scale=0.8] {$c$};
\draw (p) node[above][scale=0.8] {$p$};
\draw (a) node[left][scale=0.8] {$a$};
\draw (b) node[left][scale=0.8] {$b$};
\draw (-5.8,7.1) node[above][scale=0.8] {$z$};
\draw (-7.2,4.6) node[left][scale=0.8] {$(1+\epsilon)^i$};

\end{tikzpicture}
\caption{Illustration  for  \Cref{claim:LCA}}
\label{fig:lca}
\end{figure}
Let $x$ be an endpoint of an edge returned by the maximiser such that it lies in the first segment of $P[p,t]$.
 We first show that $x$ cannot lie on the primary path $pt$ and should necessarily lie on the secondary path.

\begin{lemma}
\label{lem:notinpz}
Let $x$ be an endpoint of an edge returned by the maximiser (\ref{eqn:6}) such that $x$ lies in the first segment of $P[p,t]$.   Let $z$ be the least common ancestor of $e_1$ and $e_2$ in the path $pe_1$ and $pe_2 \DIA e_1$. Then, $x$ cannot lie in $pz$, but it lies on $ze_2\DIA e_1$.
\end{lemma}
\begin{proof}
Let $i$ be the greatest integer such that $(1+\EP)^i$ is less than $|pa|$ on the primary path $pt$. Using \Cref{claim:LCA}, we know that $|pz| \le (1+\EP)^i$. One of the conditions of the maximiser is that $(1+\EP)^i$ distance from $p$ on the primary path remains intact. Thus, $x$ cannot lie on  $sz$. Since $x \in P$ and the detour starts on the secondary path,  $x$  lies on $ze_2 \DIA e_1$.
\end{proof}

 Let $\{e_1^*,e_2^*\}$ be the edges the maximiser returns. Without loss of generality, let us assume that $e_1^*$ lies on the primary path. Using the above lemma, $x$ cannot be an endpoint of $e_1^*$. This implies that $x$ lies on $pt \DIA e_1^*$ or the secondary path. However, we have a condition for the secondary path in our maximiser.
\begin{figure}[hpt!]

\centering
\begin{tikzpicture}[scale=0.6]
\coordinate (p) at (-4,8);
\coordinate (a) at (-7,4);
\coordinate (c) at (-3,4);
\coordinate (d) at (-2.5,3.3);
\coordinate (b) at (-7,3);
\coordinate (z) at (-5.8,7);
\coordinate (y) at (-5,7.5);
\coordinate (x) at (-5.5,5.5);

\draw[ line width=0.7mm,black,opacity=1] (p) ..node[above, rotate=70, scale=0.8] {\(pe_1\)} controls (-8,6) and (-7,6) .. (a);
\draw[ line width=0.7mm,black,opacity=1] (p)  ..node[pos=0.8,above, rotate=-30, scale=0.8] {\(pe_2 \DIA e_1\)} controls (-8,6) and (-6,6) .. (c);
\filldraw [black] (p) circle (3pt);
\filldraw [black] (z) circle (3.5pt);
\filldraw [red] (a) circle (3pt);
\draw[ line width=0.7mm,black,opacity=1] (-7.2,4.6) -- (-6.85,4.6);

\filldraw [red] (c) circle (3pt);
\filldraw [red] (d) circle (3pt);
\filldraw [red] (b) circle (3pt);

\draw[ line width=0.7mm,red,opacity=1] (c) -- (d);
\draw[ line width=0.7mm,red,opacity=1] (a) -- (b);
%\draw[ densely dashed, line width=0.7mm,blue,opacity=1] (x) .. controls (-3,5) and (-5,5) .. (c);
\filldraw [red] (c) circle (3pt);

\draw (c) node[right][scale=0.8] {$c$};
\draw (d) node[right][scale=0.8] {$d$};
\draw (p) node[above][scale=0.8] {$p$};
\draw (a) node[left][scale=0.8] {$a$};
\draw (b) node[left][scale=0.8] {$b$};
\draw (-5.8,7.1) node[above][scale=0.8] {$z$};
\draw (-7.2,4.6) node[left][scale=0.8] {$(1+\epsilon)^i$};

\filldraw [black] (x) circle (3pt);
\filldraw [black] (y) circle (3pt);

\draw (c) node[right][scale=0.8] {$c$};
\draw (-5.5,5.4) node[left][scale=0.8] {$x$};
\draw (-4.3,7.6) node[below][scale=0.8] {$y \in \LL_l$};

\filldraw [red] (c) circle (3pt);

\end{tikzpicture}
\caption{The vertex $y$ is $D$-close}
\label{fig:b46}
\end{figure}
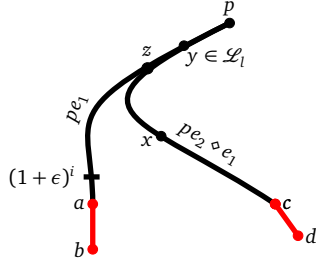
Let $\EL =  \lfloor\log |pe_2 \DIA e_1|\rfloor$.  One of the conditions of the maximiser is that $2^{\EL}$ distance from $s$ on the secondary path is intact from failures.   Thus, $|px \DIA e_1^*| \ge 2^{\EL}$. But $x$ lies on the first segment of $P[p,t]$. So $|px \DIA e_1^*| = |px| \ge 2^{\EL}$. Now, using $x$, we will find a vertex that is {\em close} to $e_2$.    Using \Cref{lem:highprobability}, with a high probability, we will find a vertex, say $y$, in path $px$ from $x$ that lies in $\LL_{\EL}$ (See \Cref{fig:b46}). Since $|pe_2\DIA e_1| = O(2^{\EL})$, even $|ye_2 \DIA e_1| =O(2^{\EL})$. Thus, $y$ is close to $e_2$. We will now show that it is $D$-close to $e_2$. To this end, we note that the detour of $P[p,t]$ starts on the secondary path after the vertex $x$ (as $x$ is intact from failures). Thus, the detour also starts after the vertex $y$ on the path $P[p,t]$. Thus, $y$ is a $D$-close vertex to $e_2$.

Remember that $p$ satisfied Lemma \ref{claim:trap}. We claim that once $a$ is outside the trapezoid of path $P[p,z]$, it is outside the trapezoid of any of its suffix paths. Specifically, for $P[y,t]$, $a$ lies outside its trapezoid. The reader can verify that Lemma \ref{claim:trap} holds even if we replace  $p$ with $y$. Thus, Lemma \ref{claim:trap} and \ref{claim:LCA} hold even for $p[y,t]$ path.

\section{From $D$-close vertex to  a clean vertex}
\label{sec:getclean}

Let us describe the setting first. In the previous section, we have found a $D$-close vertex from the source $s$. We will abuse notation and assume that the vertex $s$ itself is a $D$-close vertex and  $s \in \LL_{\EL}$. We now show how to find a clean vertex using $s$.   There are two cases, as in the previous section.

\subsection{The detour starts on the primary path}
\label{sec6:startprimary}
 In this case, we will use another maximiser similar to maximiser (\ref{eqn:3}). First, we set $dist[0] = |se_1|$. The conditions of the maximiser are as follows:

\begin{align*}
\DD_{1\alpha}(s,t,dist,clean) =  \arg\max_{F \in E^2}\{& \text{$dist[0]$ distance from $s$ on primary path is intact from $F$ ;} \\
& \text{Other conditions based on $\alpha$ from the $t$ side.}\}
 \numberthis \label{eqn:7}
\end{align*}

Let $x$ be an endpoint of an edge returned by the maximiser such that $x$ lies on the first segment of $P$. We will now show that $x$ is an $s$-clean vertex. We claim that \Cref{claim:xinst} also applies for maximiser (\ref{eqn:7}). Indeed, in this lemma, we are only using the fact that the detour starts on the primary path -- which is valid for this case too. Now, since $x$ lies on the first segment and $x$ cannot lie on the $se_1$ path (by the condition of the maximiser), $x$ cannot lie on the $st$ path. Hence, using \Cref{claim:xinst}, $x$ is an $s$-clean vertex.

\subsection{The detour starts on the secondary path}

Let us again reiterate the setting that leads us to a $D$-close vertex at the end of \Cref{sec:startsecondary}. Again, we will abuse notation and assume that $s$ is a $D$-close vertex to $e_2$. Also,  \Cref{claim:trap} holds for $s$. So, $a$ is outside the trapezoid of $P[s,t]$ path. Using this property, we can say that \Cref{claim:LCA} also holds for $s$. Let $i$ be the largest integer such that $(1+\EP)^i \le |se_1|$. Let $dist[0] = |se_2 \DIA e_1|$. The conditions of our maximiser are as follows (it is the same as maximiser (\ref{eqn:6})):

\begin{align*}
\DD_{1\alpha}(s,t,dist,clean) =  & \arg\max_{F \in E^2} \{   \text{ $(1+\EP)^i$ distance from $s$ on primary path and $dist[0]$ distance from $s$}\\ & \text{ on secondary path is intact from $F$ ; other conditions based on $\alpha$ from $t$ side.} 
 \numberthis \label{eqn:8}
\end{align*}

Let $x$ be an endpoint of an edge returned by the maximiser such that $x$ lies on the first segment of $P$.
We will now show that $x$  is an $s$-clean vertex.
 To prove our claim, we will show  that neither $V(e_1)$ nor $V(e_2)$ intersects with $T_s(x)$ as we have already assumed that $sx$ is intact from failures $\{e_1,e_2\}$.

 The proof that $V(e_1) \cap T_s(x) = \emptyset $ follows from arguments similar to \Cref{lem:notinpz}. Let us first show that $V(e_1) \cap  T_s(x) = \emptyset $. Let $z$ be the least common ancestor of $e_1$ and $e_2$ in the paths $se_1$ and $se_2 \DIA e_1$. Since $s$ satisfies \Cref{claim:LCA},  $|sz| \leq (1+ \epsilon)^i$. The first condition of the maximiser is that $(1+\EP)^i$ distance from $s$ on the primary path is intact from failures. Hence, $x$ cannot lie on the $sz$ path. It must lie on $ze_2 \DIA e_1$. Thus,    $T_s(a) \cap T_s(x) = \emptyset $.  So, $V(e_1) \cap T_s(x) = \emptyset $. Let $\{e_1^*,e_2^*\}$ be the edges returned by the maximiser. Without loss of generality, let us assume that $e_1^*$ lies on the primary path. Using similar arguments like \Cref{lem:notinpz}, $x$ is not an endpoint of $e_1^*$ as $x$ does not lie on the primary path. This implies that $x$ lies on $st \DIA e_1^*$ or the secondary path. However, we have a condition for the secondary path in our maximiser.

Our second condition is that $|se_2 \DIA e_1|$ distance from $s$ on the secondary path is intact from failures. Thus, $|sx| \ge  |se_2 \DIA e_1|$. Or, $|sx| \ge |se_2| = |sc|$. Thus, $T_s(x)$
cannot contain the vertex $c$.

We will now show that $d \notin T_s(x)$. Our argument here is the same as in the proof of \Cref{claim:xinst}. We repeat the arguments here for the sake of completeness. If $d \in T_s(x)$, then since $c \notin T_s(x)$, there exists some $c'$ such that $c'd$ is an edge of $T_s(x)$. Consider the path from $s$ to $x$ concatenated with the path from $x$ to $d$ in $T_s(x)$. This concatenated path represents the shortest path from $s$ to $d$, and it also avoids $e_1$. Additionally, it should be noted that $T_s$ was constructed by adding small weights to each edge. Therefore,   $|sd \DIA  e_1| = |sd|$ is unique and does not pass through $e_2$. Furthermore, it should be remembered that we found $st \DIA  e_1$ in the same weighted graph, and thus $st \DIA e_1$ is also unique and passes through $e_2$. However, according to the property of shortest paths, $sd \DIA e_1$ should be a subpath of $st \DIA e_1$. This leads to a contradiction, implying that $d \notin T_s(x)$. Thus, once we find a $D$-close vertex, the next vertex we can find on the first segment of $P$  has to be an $s$-clean vertex.

\section{From a clean vertex to an intermediate vertex}
\label{sec:getint}

We will assume that we have found an $s$-clean vertex, say  $p$, such that the path $P$ passes through $p$ using the process discussed in the previous section. In this section, we will show how to use a maximiser to find an intermediate vertex. First set $clean[0] = p$.  The conditions of the maximiser are as follows:

\begin{align*}
\DD_{2\alpha}(s,t,dist,clean) =  \arg\max_{F \in E^2}\{&   \text{$clean[0]$ is $s$-clean from $F$, other conditions based on $\alpha$ from $t$ side.}\}
 \numberthis \label{eqn:9}
\end{align*}

We claim that none of the endpoints of the edges returned by the maximiser is in the first segment of \(P\). For contradiction, let \(x\) be the vertex returned by the maximiser, which is in the first segment of \(P\). Then, the path \(sx\) remains intact from faults $\{e_1,e_2\}$. Since $p$ is $s$-clean, the path \(sp\) is intact from faults $\{e_1,e_2\}$. Also, both \(p\) and \(x\) lie on the path \(P\). Therefore, either \(sp\) contains \(x\) or \(sx\) contains \(p\). However, if \(x \in sp\), it contradicts our condition in the maximiser that \(sp\) is intact from faults. Thus, this possibility cannot arise. Similarly, if \(sx\) contains \(p\), then \(x \in T_s(p)\). However, this violates the condition in the maximiser that \(T_s(p)\) is intact from faults. As both possibilities lead to a contradiction, our assumption is false. Hence, \(x\) cannot lie in the first segment of \(P\). Now, if $x$ also avoids the last segment, \(x\) becomes an intermediate vertex, and we are done.

\section{Space required by our data-structures}
\label{sec:space}
Consider one of the maximisers of \cite{DuanR22}.
$$\DD_{01}(s,t, x,y) = \text{argmax}_{\{e_1,e_2\} \in E^2}\left\{ \text{$sx$ is intact from  faults $\{e_1,e_2$\} and $y$ is $t$-clean from $\{e_1,e_2\}$ } \right\}
$$

To calculate the size of the data structure, we can consider two terms: the first term represents combinations of the source vertex $s$ and the vertex $x$ (referred to as a helper vertex in \cite{DuanR22}). The number of combinations of sources and helper vertices can be $O(n^2)$. Similarly, the number of combinations of the destinations (such as $t$) and the helper vertices (such as $y$) can also be $O(n^2)$. Multiplying these two terms gives us a total space of $O(n^4)$. We now formally define this notion which will help us to calculate the space taken by our algorithm.
\begin{definition}(Space requirement from the source side)

Let us assume that the maximiser function has $k$ parameters completely dependent on the source. The space required from the source side is the number of possible combinations of these $k$ parameters and the source. Similarly, we can define the space requirement from the destination side.\end{definition}

We will primarily determine the required space from the source side. In most cases, this space will be $\TL(n)$. The calculation is symmetrical from the destination side, resulting in a total space bound of $\TL(n^2)$. However, this approach is only possible for a few maximisers. Consequently, we divide our maximisers into two groups. The space taken by the first group can be calculated using the approach mentioned above, while calculating the space requirement of the second group requires more complex techniques. Here is the partition:

\begin{enumerate}
 \item $\DD_{\alpha \beta}$ where both $\alpha, \beta  \in [0,1]$ and 
 \item  $\DD_{\alpha\beta}$ where $\alpha $ and/or $\beta$ is equal to 2.

\end{enumerate}
For the maximisers in group (1), the values of $clean[0]$ and $clean[1]$ are not set. These maximisers are used when neither an $s$-clean nor a $t$-clean vertex has been found. On the other hand, for the maximisers in the group (2), the value of $clean[0]$ and/or $clean[1]$ is set. Let us first bound the size of the group (1) maximisers.

\subsection{Group 1 maximisers}
There are two kinds of maximisers in this group: the first one in which we have not found a $D$-close vertex from $s$ and the second in which we have found the $D$-close vertex. We will bound the size of these maximisers separately.

\begin{enumerate}

\item $\DD_{0\alpha}$ where $\alpha \in \{0,1\}$

There are four maximisers of type $\DD_{0\alpha}$, namely maximiser (\ref{eqn:3}), (\ref{eqn:4}), (\ref{eqn:5}), and (\ref{eqn:6}). Among these maximisers, three of them have at most $\log n$ values set for $dist[0]$, where $dist[0]$  is:

\begin{itemize}
\item $2^{\lfloor \log |se_1| \rfloor}$ in maximiser (\ref{eqn:3})
\item $2^{\lfloor \log (\min{{|se_1|,|se_2 \DIA e_1|}})} \rfloor$ in maximiser (\ref{eqn:4})
\item $2^{\lfloor \log |pe_2|\rfloor}$ in maximiser (\ref{eqn:6})
\end{itemize}

For maximisers (\ref{eqn:3}) and (\ref{eqn:4}), $dist[0]$ has $\log n$ possibilities and the number of sources can be at most $n$. Thus, the space these maximisers take from the source side is $\TL(n)$. In the case of maximiser (\ref{eqn:6}), there is an additional condition that the $(1+\EP)^i$ distance on the primary path remains intact from failures. Again, $i$ can take $\log n$ values (assuming $\EP$ is a constant). Hence, the space taken by maximiser (\ref{eqn:6}) from the source side is also $\TL(n)$.

Now, let us consider maximiser (\ref{eqn:5}). In this maximiser, $y$ is a {\em close} vertex to $e_1$. Let us assume that  $y \in \LL_{\EL}$. We set $dist[0] = |ye_1|$. Since $y$ is close to $e_1$, the number of possible values for $dist[0]$ is $O(2^{\EL})$. Thus, the total space taken by this maximiser from the source side is $\sum_{\EL=0}^{\log n} \TL\left(\frac{n}{2^{\EL}} \times 2^{\EL}\right) = \TL(n)$.

\item $\DD_{1\alpha}$ where $\alpha \in \{0,1\}$.

There are four maximisers of type $\DD_{1\alpha}$, namely maximiser (\ref{eqn:7}) and (\ref{eqn:8}). We use these maximisers after we find a $D$-close vertex from the source side. For simplicity, let us assume that the source vertex $s$ itself is $D$-close. We further assume that $s \in \LL_{\EL}$. Since $s$ is $D$-close to the fault, we can set $O(2^{\EL})$ different values for $dist[0]$ in maximiser (\ref{eqn:7}) and (\ref{eqn:8}).  Additionally, similar to maximiser (\ref{eqn:5}), we can set at most $\log n$ values for $i$ in maximiser (\ref{eqn:8}). Thus, the space these two maximisers take from the source side is $\TL(n)$.

\end{enumerate}

\begin{comment}\textbf{Number of $\DD_{00}$ :}  Similarly, $dist[1]$ has $\log n$ many possibilities. So, there are at most $\TL(n^2)$ possibilities for $\DD_{00}$.

\textbf{Number of $\DD_{10}$ :} In this case, we assume we have found some $D-close$ vertex from the $s$ side. Let us fix some $s \in \LL_i$. For that particular $s$, we have $2^i$ many values of $dist[0]$. Also, for $t$ and $dist[1]$, like the previous part, we have $n \log n $ possibilities. Now, $|\LL_i|=\frac{n}{2^i}$. So, total possibilities of $\DD_{10}$ for vertices in $\LL_i$ is $ \frac{n}{2^i} \cdot 2^i \cdot n \log n=n^2 \log n  $. Also, values of $i$ varies from $1$ to $\log n$. So, in total, the possibilities are $n^2 \log^2 n$. Hence, the total number of data structures of type $\DD_{10}$ is $\TL(n^2)$.

\textbf{Number of $\DD_{01}$ :} By similar arguments like $\DD_{10}$, the number of data-structures of type $\DD_{01}$ is also $\TL(n^2)$.

\textbf{Number of $\DD_{11}$ :} In this case, we assume we have found some $D-close$ vertex from both the $s$ and $t$ sides. Fixing some $s \in \LL_i$ and $t \in \LL_j$, we can have $2^i$ and $2^j$ many possibilities for $dist[0]$ and $dist[1]$ respectively. Hence, total number will be $\Sigma_{j=0}^n \Sigma_{i=0}^n \Sigma_{s \in \LL_i} \Sigma_{t \in \LL_j}2^i 2^j=\Sigma_{j=0}^n \Sigma_{i=0}^n \frac{n}{2^i} \frac{n}{2^j} 2^i 2^j= \Sigma_{j=0}^n \Sigma_{i=0}^n n^2=n^2  log^2n$ which is $\TL(n^2)$.
\end{comment}

Thus, the number of maximisers of type $\DD_{00},\DD_{01},\DD_{10}$ and $\DD_{11}$ is $\TL(n^2)$.  

\subsection{Group 2 maximisers}

In a Group 2 maximiser, we set the values of $clean[0]$ and/or $clean[1]$. One might assume that there could be $O(n)$ possibilities for $clean[0]$ and $clean[1]$, implying that the size of Group 2 maximisers should be strictly greater than $\TL(n^2)$. However, this is not the case. Firstly, we claim that the number of possible vertices in $clean[0]$ is limited. For instance, let us consider the maximiser $\DD_{20}(s,t,dist,clean)$. In this maximiser, we set the value of an $s$-clean vertex as $clean[0]$. However, where does this $s$-clean vertex come from? A keen reader can observe that we must have found this $s$-clean vertex in a previous call to either the maximiser $\DD_{00}(s,t,\cdot,\cdot)$ or $\DD_{10}(s,t,\cdot,\cdot)$. 

We have already bounded the size of $\DD_{00}$ and $\DD_{10}$ maximisers, and each maximiser returns at most four vertices. Thus, we conclude that there are at most four possible values for $clean[0]$ in $\DD_{20}(s,t,dist,clean)$ for each maximiser of form $\DD_{00}(s,t,\cdot,\cdot)$ or $\DD_{10}(s,t,\cdot,\cdot)$. Thus, the number of possible vertices in $clean[0]$ is $\TL(n^2)$. 

Using the same strategy, we can bound the space taken by Group 2 maximisers. There is only one Group 2 maximiser, that is,  maximiser (\ref{eqn:9}). This maximiser can be of three types: $\DD_{20}, \DD_{21}$ and $\DD_{22}$. We will individually bound the space for each type.

\begin{enumerate}
\item $\DD_{20}$

In this case, we have already found an $s$-clean vertex from the source side. As mentioned earlier, this $s$-clean vertex must be the output of a maximiser, specifically of type $\DD_{00}$ or $\DD_{10}$. Let's assume that we have found an $s$-clean vertex $x$ from the maximiser $\DD_{00}(s,t,dist',clean')$. Two changes are made in the subsequent call to the maximiser $\DD_{20}(s,t,dist,clean)$. First, we set $bits[0]=2$, which renders $dist[0]$ irrelevant (as in maximiser  (\ref{eqn:9}), $dist[0]$ is not used in any condition). The main change is setting $clean[0]$ to $x$. Since each maximiser outputs two edges or four vertices, for each maximiser of type $\DD_{00}$ or $\DD_{10}$, at most four maximisers of type $\DD_{20}$ are created. Hence, the total size of $\DD_{20}$ is $\TL(n^2)$. Similarly, the size of $\DD_{02}$ is also $\TL(n^2)$.

\item $\DD_{21}$

Once again, let us consider the scenario that occurred in our algorithm just before utilising this data structure. At that point, we either discovered an $s$-clean vertex or identified a $D$-close vertex from $t$. Let us focus on the more straightforward case we have already addressed earlier.

\begin{enumerate}
\item We have found an $s$-clean vertex.

The $s$-clean vertex obtained must have been the output of either the maximiser $\DD_{01}$ or $\DD_{11}$. Similar to the previous case, we assert that the overall size of such a data structure is $\TL(n^2)$.

\item We have found a $D$-close vertex from $t$.

This $D$-close vertex must be the output of the maximiser $\DD_{20}$. Let $t$ be a $D$-close vertex found using the maximiser  $\DD_{20}(s,t',dist',clean')$. After finding $t$, our algorithm  uses the maximiser $\DD_{21}(s,t,dist,clean)$. The pair $s, clean[0]$ remains unchanged between $\DD_{20}$ and $\DD_{21}$. The values of $dist[0]$ and $dist'[0]$ are not used in the conditions of maximisers $\DD_{21}$ and $\DD_{20}$, respectively, so they remain unchanged.

The only change is that $t'$ in $\DD_{20}$ is replaced by the $D$-close vertex $t$, and $dist[1]$ is set to a suitable value. Let $t \in \LL_{\EL}$. The number of possible values for $dist[1]$ is $O(2^{\EL})$. Thus, the total space taken from the destination side in the maximiser $\DD_{21}$ is $\sum_{\EL = 0}^{\log n} \frac{n}{2^{\EL}} 2^{\EL} = \TL(n)$. 

We will now examine the total space from the source side and show it as $\TL(n)$. As stated above,  there is no change in $s$ and $clean[0]$ between the maximisers $\DD_{20}$ and $\DD_{21}$. Thus, the space taken by the maximiser $\DD_{21}$ from the source side is the same as the space taken by the maximiser $\DD_{20}$ from the source side. We now bound the space taken by $\DD_{20}$ from the source side. 

We use the maximiser $\DD_{20}$ after we obtain an $s$-clean vertex either from $\DD_{00}$ or $\DD_{10}$. We look at these two transitions to bound the space of $\DD_{20}$ from the source side.
Let us go over these two transitions separately:

\begin{itemize}
\item  The $s$-clean vertex in the maximiser $\DD_{20}$ is obtained from  the maximiser $\DD_{00}$.

Consider the maximiser $\DD_{00}(s, t, dist, clean)$. We assume that we have found an $s$-clean vertex from this maximiser, which is then used in the subsequent call to maximiser $\DD_{20}$. By applying \Cref{rem:notusecondmaximiser},  this $s$-clean vertex cannot be the output of maximiser (\ref{eqn:5}).

Therefore, the maximiser $\DD_{00}$ can fall into types (\ref{eqn:3}), (\ref{eqn:4}), and (\ref{eqn:6}). In these maximisers, there can be at most $\log n$ different values for $dist[0]$ and $dist[1]$. Consequently, at most $O(\log^2 n)$ clean vertices can serve as the output of the maximiser $\DD_{00}$. As a result, the total number of possible combinations of source and $clean[0]$ in the maximiser $\DD_{20}$ is bounded by $O(n \log^2 n)$. Hence, the space required from the source side in the maximiser $\DD_{20}$ is $\TL(n)$.
\item The $s$-clean vertex in the maximiser  $\DD_{20}$ is obtained form the maximiser  $\DD_{10}$.

Consider the maximiser $\DD_{10}(s, t, dist, clean)$. We assume that we have found an $s$-clean vertex from this maximiser, which is then used in the subsequent call to the maximiser $\DD_{20}$. Again, by using \Cref{rem:notusecondmaximiser}, this $s$-clean vertex cannot be the output of the maximiser \ref{eqn:5}.

If the maximiser $\DD_{10}$ falls into types (\ref{eqn:3}), (\ref{eqn:4}), or (\ref{eqn:6}), then there can be at most $\log n$ different values of $dist[1]$. By the definition of $\DD_{10}$, $s$ is a $D$-close vertex from the source. Let us assume that $s \in \LL_{\EL}$. Consequently, the number of possible different values of $dist[0]$ is at most $O(2^{\EL})$.

As a result, at most $O(2^{\EL}\log n)$ clean vertices can serve as the output of the maximiser $\DD_{10}$. Thus, the total number of possible combinations of source and $clean[0]$ in $\DD_{20}$ is $\sum_{\EL=0}^{\log n} O\left(\frac{n}{2^{\EL}} 2^{\EL}\log n\right)$. Therefore, the space required from the source side in $\DD_{20}$ is $\TL(n)$.
\end{itemize}
Thus, the total space taken by $\DD_{21}$ is $\TL(n^2)$. The same argument also applies to $\DD_{12}$.

\end{enumerate}
\item $\DD_{22}$

This data structure is invoked when we obtain either an $s$-clean or a $t$-clean vertex from $\DD_{20}$, $\DD_{21}$, $\DD_{02}$, or $\DD_{21}$. All of these data structures have a size of $\TL(n^2)$. As each maximiser outputs two edges or four vertices, the size of $\DD_{22}$ is also $\TL(n^2)$.
\end{enumerate}

\section{Time taken by the $\QU$ algorithm}
\label{sec:time}

First, we note that it takes $O(1)$ time to determine all the parameters of the maximiser. From the source side, to set $dist[0]$, we need to find  $|se_1|$ and $|se_2 \DIA e_1|$. Since we do not know whether the detour starts or ends on the primary or secondary path, we exhaustively check all possibilities. There are only four possible cases:

\begin{enumerate}
\item The detour starts and ends on the primary path.
\item The detour starts on the secondary path but ends on the primary path.
\item The detour starts on the primary path but ends on the secondary path.
\item The detour starts and ends on the secondary path.
\end{enumerate}

Let us first consider the most straightforward case: when the detour starts and ends on the primary path. Discussing the running time in the context of the flowchart presented in \Cref{fig:ourapproach} is easier. After the invocation of the first maximiser, namely maximiser (\ref{eqn:3}), in the worst case, we can obtain a $D$-close vertex from $s$. The second invocation may give a $D$-close vertex from $t$, and so on. We claim that we must have found the intermediate vertex after four invocations of maximisers. Since each maximiser returns two edges or four vertices, there are at most 16 vertices to be processed. This can be done in $O(1)$ time. Thus, we obtain the intermediate vertex in $O(1)$ time.

The analysis is similar for the second case: when the detour starts on the secondary path but ends on the primary path. In this case, finding a $D$-close vertex from the source vertex itself may require three invocations of maximisers (see \Cref{sec:startsecondary}). Thus, the total number of calls to the maximiser increases to six. Even in this case, the running time is $O(1)$.

The last two cases are similar to the previous two cases. Since we do not know which case we are in, we run our algorithm for all four cases. Each case takes $O(1)$ time, after which we obtain $O(1)$ candidates for the intermediate vertices. Simultaneously, we always maintain an upper bound on the length of $P = |st \DIA \{e_1,e_2\}|$. Using \Cref{lem:notinpathP}, if the edges returned by any of our maximisers do not intersect with $P$, we can update our upper bound on the size of $P$ (represented by the value $L$ returned by the $\HIT$ algorithm). Therefore, our $\HIT$ algorithm returns an upper bound on $P$ and a set $H$ of $O(1)$ candidates for the intermediate vertices. Following the analysis in \Cref{sec:overview}, the running time of our algorithm is $O(1)$.   

\section{Both the faults are on the primary path}
\label{both}Let us first define the setting in this easy case when both the faulty edges lie on the primary path. We want to find the length of  $P=st \DIA \{e_1,e_2\}$ where  $e_1 = (a,b)$ and $e_2=(c,d)$ lie on the primary path. We further assume that $a$ is closer to $s$ and $d$ is closer to $t$. 

In \Cref{sec:twofault}, we designed an algorithm for the case when the faults were on the primary as well as the secondary path. One of the corner cases in that section was as follows: We find a $D$-close vertex, say $y$, such that the primary path $yt$ contains both faults. Thus, we jump from one case to the other. We will now show that we cannot jump in the reverse direction. To this end, we state the following lemma without proof. The proof of this lemma is trivial.

\begin{lemma}
\label{lem:remain} 
If the path $st$ has both the faults, then for any $y \in se_1$,  both the faults are in the primary path $yt$. Similarly, for $y \in te_2$ path, both the faults are in the primary path $sy$. 
\end{lemma}

In the case we are dealing with in this section, the reader will see that if we find a $D$-close vertex, say $p$, from $s$, then it will lie on $se_1$ path. Thus, using the above lemma, we will remain in the case where both the faults lie on the primary path $pt$. Thus, we will not jump from this case to any other case. Let us now discuss how we handle this case.  

The reader will see that this case is nearly similar to the case described in \Cref{sec:startprimary}.
Indeed, it is necessarily the case here that the detour starts and ends on the primary path. Our first maximiser is the same as the maximiser (\ref{eqn:3}). We set $dist[0] = 2^{\log \lfloor |se_1| \rfloor}$. The conditions of the maximiser are as follows:

 \begin{align*}
\DD_{0\alpha}(s,t,dist,clean) =  \arg\max_{F \in E^2}\Biggl\{& dist[0]  \text{ distance from $s$ on the primary path is intact from $F$ and}\\ %& \text{ $dist[0] = O(2^{\lfloor\log |se_1|\rfloor})$ and} \\
& \text{Other conditions based on $\alpha$ from the $t$ side.}\Biggl\}
 \numberthis \label{eqn:10}
 \end{align*}

The only difference between maximiser (\ref{eqn:10}) and maximiser (\ref{eqn:3}) is the common condition of (\ref{eq:common}). For our case, the common condition is as follows:
\begin{equation}
\text{Common Condition: A pair of edges $(e,e')$ will be considered in the maximiser only if $e,e' \in st$}
\end{equation}

  Let $x$ be an endpoint of an edge returned by the maximiser such that it lies on the first segment of $P$. Now,  like \Cref{claim:xinst}, we claim that either $x$ lies on the $st$ path or is $s$-clean. The proof of this is easier than the proof of \Cref{claim:xinst} but follows similar arguments.

  \begin{lemma}[Same as \Cref{claim:xinst}]
 \label{claim:xsclean}
        Let $x$ be an endpoint of an edge returned by the maximiser $\DD_{0 \alpha}$ such that $sx$ lies on the first segment of $P$. If $x$ does not lie on the $st$ path, then $x$ is an $s$-clean vertex.
  \end{lemma}

   \begin{proof}
        Let us assume that $x \notin st$. By the condition in the lemma, $sx$ is intact from failures. Thus, to show that $x$ is an $s$-clean vertex, we only need to show that the vertices in $V(e_1)$ and $V(e_2)$ do not lie in $T_s(x)$.

        Since $x$ is not in the $st$ path and hence not in the $sa$ path,  the two subtrees $T_s(x)$ and $T_s(a)$ are disjoint. Here, $e_1, e_2 \in T_s(a)$ (as both the edges are on the primary path itself). So, $(V(e_1) \cup V(e_2)) \cap
T_s(x) = \emptyset$.
   \end{proof}

If $x$  lies on the $st$ path then, similar to \Cref{sec:startprimary}, we find a $D$-close vertex to $e_1$ using $x$, say $y$. Here, note that, by construction, $y$ lies on $se_1$ path. Using \Cref{lem:remain}, the primary path $yt$ also contains faults $\{e_1,e_2\}$.  Then, using a maximiser similar to the one used in  \Cref{sec6:startprimary}, we can find an $s$-clean vertex. Similarly, we get an intermediate vertex using a maximiser similar to the one used in  \Cref{sec:getint}. Since our arguments are similar to the arguments in \Cref{sec6:startprimary} and \Cref{sec:getint}, we do not repeat them here. The space and the time taken in this case also can be bounded using the arguments presented in \Cref{sec:space} and \Cref{sec:time}, respectively. Thus, we have proven the main result of the paper, that is \Cref{thm:main}.

\bibliographystyle{plainurl}
\bibliography{paper.bib}
%-Bibliography
\appendix

\end{document}